\newcommand*{\FULL}{}%
\newcommand*{\WORLDCUP}{}%

\documentclass{vldb}
\ifdefined\FULL
\else
\usepackage{nopageno}
\fi
\usepackage{epsfig}
\usepackage{graphicx}
\usepackage{color}
\usepackage{fullpage}
\usepackage{amssymb}
\usepackage{amsmath}
\usepackage{theorem}
\usepackage[linesnumbered,ruled,vlined]{algorithm2e}
\usepackage{subfig}
\usepackage{float}
\usepackage{afterpage}
\usepackage{url}

\newtheorem{theorem}{Theorem}
\newtheorem{lemma}{Lemma}

\newtheorem{corollary}{Corollary}

\newtheorem{definition}{Definition}
\newtheorem{remark}{Remark}

\newcommand{\abs}[1]{\left| #1 \right|}
\newcommand{\norm}[1]{\left\lVert #1 \right\rVert}
\newcommand{\E}{\mathbf{E}}
\newcommand{\var}{\mathbf{Var}}
\newcommand{\err}[1]{\text{Err}_{#1}^k}

\newcommand{\bx}{\mathbf{x}}
\newcommand{\bz}{\mathbf{z}}
\newcommand{\by}{\mathbf{y}}
\newcommand{\bw}{\mathbf{w}}

\newcommand{\bpi}{\boldsymbol{\pi}}

\newcommand{\infnorm}[1]{\|#1\|_{\infty}}

\renewcommand{\Pr}{\mathbf{Pr}}

\renewcommand{\S}{\mathcal{S}}

\newcommand{\CS}{\textsf{CS}}

\newcommand{\CM}{\textsf{CM}}
\newcommand{\CMCU}{\textsf{CM-CU}}
\newcommand{\CMLCU}{\textsf{CML-CU}}
\newcommand{\OMP}{\textsf{OMP}}
\newcommand{\BOMP}{\textsf{BOMP}}
\newcommand{\oneSR}{\textsf{$\ell_1$-S/R}}
\newcommand{\twoSR}{\textsf{$\ell_2$-S/R}}

\newcommand{\median}{\operatornamewithlimits{median}}
\newcommand{\mean}{\operatornamewithlimits{mean}}
\newcommand{\argmin}{\operatornamewithlimits{argmin}}

\renewcommand{\paragraph}[1]{\medskip \noindent {\bf #1.}}

\begin{document}

\title{Bias-Aware Sketches}

\numberofauthors{2} 

\author{
\alignauthor
Jiecao Chen\\
\affaddr{Indiana University}\\
\affaddr{Bloomington, IN 47405}\\
\email{jiecchen@umail.iu.edu}\\
\alignauthor
Qin Zhang\\
\affaddr{Indiana University}\\
\affaddr{Bloomington, IN 47405}\\
\email{qzhangcs@indiana.edu}
}
\date{}


\maketitle




\begin{abstract}
Linear sketching algorithms have been widely used for processing large-scale distributed and streaming datasets. Their popularity is largely due to the fact that linear sketches can be naturally composed in the distributed model and be efficiently updated in the streaming model.  The errors of linear sketches are typically expressed in terms of the sum of coordinates of the input vector excluding those largest ones, or, the mass on the tail of the vector.  Thus, the precondition for these algorithms to perform well is that the mass on the tail is small, which is, however, not always the case -- in many real-world datasets the coordinates of the input vector have a {\em bias}, which will generate a large mass on the tail.  

In this paper we propose linear sketches that are {\em bias-aware}. 
We rigorously prove that they achieve strictly better error guarantees than the corresponding existing sketches, and demonstrate their practicality and superiority via an extensive experimental evaluation on both real and synthetic datasets.

\end{abstract}

\section{Introduction}
Linear sketches, such as Count-Sketch~\cite{CCFC02} and Count-Median~\cite{CM04}, are powerful tools for processing massive, distributed, and real-time datasets.   Let $\bx = (x_1, \ldots, x_n)^T$ be the input data vector where $x_i$ stands for the frequency of element $i$.  Linear sketching algorithms typically consist of two phases: (1) {\em Sketching phase.} We apply a linear sketching matrix $\Phi \in \mathbb{R}^{r \times n}\ (r \ll n)$ on $\bx$, getting a sketching vector $\Phi \bx$ whose dimension is much smaller than $\bx$. (2) {\em Recovery phase.} We use $\Phi \bx$ to recover useful  information about the input vector $\bx$, such as the median coordinate, the number of non-zero coordinates (distinct elements), etc.

We start by explaining why linear sketches are useful in handling distributed and streaming data.  In the distributed computation model, we have $t$ data vectors $\bx^1, \ldots, \bx^t$ distributed at $t$ sites, which connect to a central coordinator.  The goal is for the coordinator to learn the global data vector $\bx = \sum_{i \in [t]} \bx^i$ communication efficiently.  Note that the naive solution that each site sending $\bx^i$ to the coordinator is communication expensive if the dimension of $\bx$ is large. By linearity we have $\Phi \bx = \Phi \bx^1 + \ldots + \Phi \bx^{t}$.   Thus each site can simply send the local sketching vector $\Phi \bx^i$ to the coordinator, and then the coordinator sums up these local sketching vectors to obtain the global sketching vector $\Phi \bx$, from which it reconstructs $\bx$ using the recovery procedure.  The total communication will be the product of $t$ and the dimension of $\Phi \bx$, which is much smaller than the dimension of input vector $\bx$.  

In the streaming model~\cite{AMS96}, where items arrive one by one in the online fashion, a new incoming item $i \in [n]$ corresponds to updating the input vector $\bx \gets \bx + \mathbf{e}_i$ where $\mathbf{e}_i$ is an all-$0$ vector except the $i$-th coordinate being $1$.  Again due to linearity, we can easily update the linear sketch as $\Phi \bx \gets \Phi \bx + \Phi \mathbf{e}_i$.  The space usage of the streaming algorithm is simply the dimension of the sketch $\Phi \bx$, which is again much smaller than the dimension of  $\bx$.

We consider in this paper the basic problem that in the recovery phase, we want to best reconstruct the input vector $\bx$ using the sketching vector $\Phi \bx$.  More precisely, our goal is to design a sketching matrix $\Phi$ and a recovery procedure $\mathcal{R}(\cdot)$ with the following properties. 
\begin{itemize}
\vspace{-1mm}
\item {\em Accuracy.}  $\hat{\bx} = \mathcal{R}(\Phi \bx)$ is close to the original vector $\bx$ under certain distance measurement.

\vspace{-1mm}
\item {\em Compactness.} The size of the sketch (equivalently, $r$, the number of rows of  $\Phi$) is small; 

\vspace{-1mm}
\item {\em Efficiency.} We can compute $\Phi \bx$ and $\hat{\bx} = \mathcal{R}(\Phi \bx)$ time-efficiently.  
\end{itemize}

\vspace{-1mm}

This basic problem has many applications in massive data processing. Once a good approximation to $\bx$ is obtained, we can answer a number of statistical queries on the input frequency vector such as {\em point query}, {\em  frequent elements}, {\em range query}, etc. These queries have numerous real-world applications, including Internet data analytics \cite{CJK+04}, search engines \cite{MGL+11}, data stream mining \cite{CH10}, streaming and distributed query processing \cite{Cor11,CG05,VWW+15}, etc. 

In this paper we focus on point query, which we believe is the most basic operation: given an index $i \in [n]$, return $x_i$ (the $i$-th coordinate of the input vector $\bx$).
%
Naturally, we would like to minimize the maximum (average) coordinate-wise difference between the recovered vector $\hat{\bx} = \mathcal{R}(\Phi \bx)$ and the original vector $\bx$, that is, to minimize $\norm{\bx - \hat{\bx}}_\infty$ ($\frac{1}{n}\norm{\bx - \hat{\bx}}_1$).  


\paragraph{Linear Sketches}
Before stating our results, we would like to add some background on linear sketches. For a general vector $\bx \in \mathbb{R}^n$, it is impossible to recover $\bx$ {\em exactly} from the sketching vector $\Phi \bx$ of a much smaller dimension.  However, in many cases we are able to recover $\bx$ up to some small errors. One such error guarantee, called the $\ell_\infty/\ell_p$-guarantee, is that for any $\bx \in \mathbb{R}^n$, letting $\hat{\bx} = \mathcal{R}(\Phi \bx)$, the coordinate-wise error of the recovery is bounded by
\begin{eqnarray}
\label{eq:a-1}
\norm{\hat{\bx} - \bx}_\infty = O(k^{-1/p}) \cdot \err{p}(\bx),
\end{eqnarray} 
where $k$ is a tradeoff parameter between the sketch size and the accuracy guarantee, and $$\err{p}(\bx) = \min_{k\text{-sparse}\  \bx'}\norm{\bx - \bx'}_p,$$ where we say a vector is {\em $k$-sparse} if it contains {\em at most} $k$ non-zero coordinates. In other words, $\err{p}(\bx)$ is the $\ell_p$-norm of the vector containing all coordinates of $\bx$ except zero-ing out the $k$ coordinates with the largest absolute values. We often call the $k$ largest coordinates the {\em head} of $\bx$ and the rest $(n-k)$ ones the {\em tail} of $\bx$.  
Note that if $\bx$ is $k$-sparse, then we are able to recover it exactly since $\err{p}(\bx) = 0$.  

We typically consider $p = 1$ or $p = 2$, since for $p > 2$ there exists strong lower bound: the sketch size has to be at least $\Omega(n^{1-2/p})$.~\footnote{The proof can be done using the $n^{1/p}$-party set-disjointness hard instance similar to that for $p$-th frequency moments~\cite{BJKS02}.}
 The error guarantee in Equality (\ref{eq:a-1}) for $p = 1$ and $p = 2$ can be achieved with high probability by the classical Count-Median algorithm~\cite{CM04} and Count-Sketch algorithm~\cite{CCFC02} respectively; we will illustrate these two algorithms in details in Section~\ref{sec:preliminary}.  

It is folklore that $\ell_\infty/\ell_1$ and $\ell_\infty/\ell_2$ guarantees can be converted into $\ell_1/\ell_1$ and $\ell_2/\ell_2$ guarantees respectively (see, for example, Section II of \cite{GI10}).  More precisely, for $p \in \{1, 2\}$ we can derive from Inequality (\ref{eq:a-1}) that
\begin{eqnarray}
\label{eq:a-1-1}
\norm{\hat{\bx} - \bx}_p = O(1) \cdot \err{p}(\bx),
\end{eqnarray} 
which gives a more intuitive approximation guarantee on the whole vector instead of individual coordinates. 

\smallskip

\paragraph{Bias-Aware Sketches}
The question we try to address in this paper is:
\begin{quote}
\em What if the coordinates in the input vector $\bx$ have a non-trivial bias?
\end{quote}

Let us consider an example. Let $k = 2$, $n = 10$, and
\begin{equation}
\label{eq:example}
\bx = (\mathbf{3}, 100, 101, \mathbf{500}, 102, 98, 97, 100, 99, 103).
\end{equation}
We have $\err{1}(\bx) = 700$,  $\err{2}(\bx) = \sqrt{69428} \approx 263.49$, which are fairly large. It is easy to see that these large errors are due to the fact that most coordinates of $\bx$ are close to $100$ (intuitively, the bias), which results in a heavy tail.  It would be desirable if we can remove this bias first and then perform the sketching and recovery.  

In this paper we propose  bias-aware sketches that achieve the following performance guarantee.  Let $\beta^{(n)}$ be the $n$-dimensional vector with $\beta$ at each coordinate.   For $p \in \{1,2\}$, our sketches can recover an $\hat{\bx}$ such that
\begin{eqnarray}
\label{eq:a-2}
\textstyle \norm{\hat{\bx} - \bx}_\infty = O(k^{-1/p}) \cdot  \min_{\beta} \err{p}(\bx -\beta^{(n)}).
\end{eqnarray}
And we define the {\bf bias} of the input data vector $\bx$ to be
\begin{eqnarray}
\label{eq:a-3}
\beta^* = \arg\min_{\beta} \err{p}(\bx -\beta^{(n)}).
\end{eqnarray}

Clearly, the right hand side (RHS) of Inequality~(\ref{eq:a-2}) is no more than the RHS of Inequality~(\ref{eq:a-1}) (equal when the best bias $\beta$ is $0$). 
In the case when all except at most $k$ coordinates of $\bx$ are close to a non-zero $\beta$, our error bound will be much better than that in (\ref{eq:a-1}).  For the example mentioned earlier, we have $\min_{\beta}\err{1}(\bx - \beta^{(10)}) = 12$ and $\min_{\beta}\err{2}(\bx - \beta^{(10)}) = \sqrt{28} \approx 5.29$ ($\arg\min_\beta = 100$; in this example the bias happens to be the same for both $p = 1$ and $p = 2$), which are significantly smaller than those given by Count-Median and Count-Sketch.  

Same as Inequality~(\ref{eq:a-1-1}), for $p \in \{1, 2\}$ we can derive from (\ref{eq:a-2}) that
\begin{eqnarray}
\label{eq:a-2-1}
\textstyle \norm{\hat{\bx} - \bx}_p = O(1) \cdot  \min_{\beta} \err{p}(\bx -\beta^{(n)}).
\end{eqnarray}


\begin{remark}
\label{rem:single-bias}
Compared with the single bias $\beta$, one may want to allow multiple bias values.  For example, for the data vector $\by = (200, 100, 50, 50, 50, 50, 100, 100, 100, 10)$, one may want to use two bias values $\beta_1 = 50$ and $\beta_2 = 100$, with $200$ and $10$ being the outliers. Unfortunately, this {\em cannot} be done if we want to obtain an $o(n)$ (sublinear) size sketch where $n$ is the dimension of the input vector, simply because when we have at least two bias values, in the recovery procedure for each of the $n$ coordinates of input vector we need the information of which bias value has been deducted from that coordinate, which costs at least $1$ bit.
\end{remark}


\vspace{-1.5mm}

\paragraph{Our Contributions}  
In this paper we have made the following contributions.
\begin{enumerate}
\item We have given a rigorously formalization of the bias-aware sketches, which {\em strictly generalizes} standard linear sketches in the error guarantees.  

\item We have proposed bias-aware sketches with rigorous $\ell_\infty/\ell_1$ and $\ell_\infty/\ell_2$ error guarantees.  We have also shown how to implement our sketches in the streaming model for fast real-time query.

\item We have implemented our algorithms and verified their effectiveness on both synthetic and real-world datasets.  We note that our algorithms significantly outperform the existing algorithms in terms of accuracy for point query. 
\end{enumerate}


\section{Related Work}
\label{sec:related}
The history of data sketch/summary can be traced back to Morris' {\em approximate counter}~\cite{M78} and Flajolet and Martin's {\em probabilistic counting} algorithm~\cite{FM85}.  Subsequently, streaming algorithms were extensively investigated since the seminal paper \cite{AMS96} by Alon {\em et al.}  Among them Count-Sketch~\cite{CCFC02} and Count-Min/Count-Median~\cite{CM04} were found particularly useful in many applications from data analytics and mining to query processing and optimizations.  A number of variants of the Count-Min  algorithm have also been proposed, such as Count-Min with conservative update~\cite{EV02,GDC12} and Count-Min-Log with conservative update~\cite{PF15}, but these sketches are {\em not} linear and thus cannot be directly used  in the distributed setting.  Another closely related algorithm is the Counter-braids~\cite{LMPDK08}.  The intent of Counter-braids is to be more bit-efficient than methods which simply use counters. It requires a larger amount of space to execute; and its encoding/decoding procedures are recursive, layer by layer, and thus it cannot answer point query without decoding the whole input vector $\bx$.
Finally, we would like to emphasize that all of the algorithms mentioned above cannot handle data bias.

Deng {\em et al.}~\cite{DR07} attempted to remove the bias in the Count-Min algorithm.  In the high level, at the time of recovering a coordinate mapped to a hash bucket (see CM-matrix in Definition~\ref{def:CM-matrix}), their algorithm averages the coordinates mapped into all other hash buckets to obtain an estimate of the bias presented in the considered bucket. It turns out that such an estimation is too rough to be useful -- their analysis shows that their algorithm can only achieve comparable recovery quality as Count-Sketch.

Yan {\em et al.}~\cite{YZH+15} formulated the bias recovery problem in the context of distributed outlier detection.  We briefly describe how \BOMP\ works.  To sketch a vector $\bx \in \mathbb{R}^n$, \BOMP\ first computes $\by = \Phi \bx$ where $\Phi = [\phi_1, \ldots, \phi_n] \in \mathbb{R}^{t\times n}$, where each entry of $\Phi$ is independently sampled from the Gaussian distribution $\mathcal{N}(0, 1/t)$. In the recovery phase \BOMP\ prepends a new column $\frac{1}{\sqrt{n}}\sum_{i=1}^n \phi_i$ to $\Phi$ to get $\textstyle \Phi' = [\frac{1}{\sqrt{n}}\sum_{i=1}^n \phi_i, \Phi]$,  and then runs \OMP\ ({\em Orthogonal Matching Pursuit}) on $\by$ and $\Phi'$ in $k + 1$ iterations to recover $\tilde{\bx}$ as an approximation of $\bx$.
However, their discussion only focused on the biased $k$-sparse vectors  where all coordinates of $\bx$ are equal to some unknown value $\beta$ except at most $k$ ``outliers'', and did not give a solid theoretical analysis. 
 Moreover, \OMP\ is very time expensive, and cannot answer point query without decoding the whole vector $\bx$.


Our work is closely related to the area of compressive sensing. In fact, our linear sketching and recovery algorithms can be seen as natural extensions of the standard compressive sensing sparse recovery algorithms~\cite{CRT06,Donoho06,CM06}. 
In the standard sparse recovery setting the bias of the vector is assumed to be $0$, which does work well for a number of problems in signal processing but its power is somewhat limited for massive data processing where coordinates in vectors may have non-zero biases.   We note that the idea of debiasing can be viewed as a special case of the incoherent dictionary learning~\cite{DET06,GMS03} -- one can add an all-$1$ vector (normalized by $1/\sqrt{n}$) upon the $n$ standard basis vectors. However, as far as we are concerned, the existing recovery algorithms in incoherent dictionary learning use either linear programming or \OMP, which, again, are very time-inefficient on large datasets and do not work for point query.


\section{Preliminaries}
\label{sec:preliminary}

We summarize the main notations in this paper in Table~\ref{tab:notation}. A quick scan of the table may be useful since some of the notations are not standard (e.g., a vector minus a scalar value: $\bx - \beta$).


\begin{table}[!ht]
\centering
\caption{List of notations}
\label{tab:notation}
\begin{tabular}{|l|l|}
\hline
$[n]$ & $[n] =\{1, 2, \ldots, n\}$  \\
\hline
$\Pr$ & the probability of  \\
\hline
${(\bx)}_i$ or $x_i$ & for $\bx \in \mathbb{R}^n$, both ${(\bx)}_i$ and $x_i$
represent \\
&the $i$-th coordinate of $\bx$ \\
\hline
$\|\bx\|_p$ &  $\|\bx\|_p = (\sum_i|x_i|^p)^{\frac{1}{p}}$ for $\bx = (x_1, \ldots, x_n)$;\\
            &  when $p=\infty$, $\|\bx\|_{\infty} = \max_i|x_i|$\\
\hline
$k$-sparse  & $\bx\in\mathbb{R}^n$ is $k$-sparse if $\bx$ has \emph{at most} $k$\\ 
& non-zero coordinates\\
\hline
$\S_{m}(\bx)$ & set of vectors in $\mathbb{R}^m$ obtained by choosing \\
        & $m\ (\le n)$ coordinates from $\bx \in \mathbb{R}^n$\\
\hline
$\err{p}(\bx)$ &  $\err{p}(\bx) = \min_{k\text{-sparse}~\bx'} \|\bx - \bx'\|_p$ \\ 
\hline
$\bx - \beta$ & for $\bx \in \mathbb{R}^n, \beta \in \mathbb{R}$, \\
& $\bx - \beta = (x_1 - \beta,  \ldots, x_n -\beta)$ \\
\hline 
$\mean(\bx)$ & for $\bx \in \mathbb{R}^n$, $\mean(\bx) = \frac{1}{n}\sum_{i=1}^nx_i$\\
\hline
$\median(\bx)$ & for $\bx \in \mathbb{R}^n$,  $\median(\bx)=x_{\frac{n+1}{2}}$ for odd $n$, \\
&   $\median(\bx)=(x_\frac{n}{2} + x_{\frac{n}{2}+1})/2$ for even $n$\\
\hline
$\argmin_{\beta}f(\beta)$ & $\argmin_{\beta}f(\beta) = \{\alpha ~|~ f(\alpha) = \min_{\beta}f(\beta)\}$ \\
\hline
$\sigma^2(\bx)$ &  variance of $\bx \in\mathbb{R}^n$; \\
& $\sigma^2(\bx) = \frac{1}{n}\sum_{i=1}^n(x_i - \mean(\bx))^2$\\
\hline
$\sigma^2(Y)$ &   variance of a random variable $Y$;\\
& $\sigma^2(Y) = \E\left[(Y - \E[Y])^2\right]$\\
\hline
$\Pi$ & CM-Matrix. See Definition \ref{def:CM-matrix}\\
\hline
$\Psi$ & CS-Matrix. See Definition \ref{def:CS-matrix}\\
\hline
$\Upsilon$ & Sampling matrix. See Definition \ref{def:sample-matrix}\\
\hline
\end{tabular}
\end{table}

\smallskip

We would like to introduce two classical linear sketches Count-Median and Count-Sketch, which will be used as components in our algorithms.

\paragraph{Count-Median}
The Count-Median algorithm \cite{CM04} is a linear sketch for achieving $\ell_\infty/\ell_1$-guarantee.  We first introduce the {\em Count-Median} matrix.
\begin{definition}[CM-matrix]
\label{def:CM-matrix}
Let $h : [n] \rightarrow [s]$ be a hash function. A CM-matrix $\Pi(h)\in \{0, 1\}^{s\times n}$ is defined as
\begin{equation*}
\Pi(h)_{i,j} = \left\{
  \begin{array}{ll}
    1  & h(j) = i\\
    0  & h(j) \neq i.
  \end{array}
\right.
\end{equation*}   
\end{definition}

For a vector $\bx \in \mathbb{R}^n$, the following theorem shows that we can recover each coordinate of $\bx$ with a bounded error from $\Theta(\log n)$ random sketching vectors $\Pi(h) \bx$.

\begin{theorem}[\cite{CM04}]
\label{thm:count-median}
Set $s = \Theta(k/\alpha)$ for an $\alpha \in (0,1)$ and $d = \Theta(\log n)$.  Let $h^1, \ldots, h^d : [n] \to [s]$ be $d$ independent random hash functions, and let $\Pi(h^1), \ldots, \Pi(h^d)$ be the corresponding CM-matrices.  Let $\hat{\bx} = (\hat{x}_1, \ldots, \hat{x}_n)$ be a vector such that
$$
\hat{x}_j = \median_{i\in[d]} \left\{ \left(\Pi(h^i)\bx\right)_{h^i(j)} \right\}.
$$
We have
$
\Pr\left [\infnorm{\hat{\bx} - \bx} \leq \alpha / k \cdot \err{1}(\bx)\right] \geq 1 - 1/n.
$
\end{theorem}

\paragraph{Count-Sketch}
The Count-Sketch algorithm \cite{CCFC02} is a linear sketch for achieving $\ell_\infty/\ell_2$-guarantee. It is similar to Count-Median; the main difference is that it introduces random signs in the sketching matrix.
\begin{definition}[CS-Matrix]
\label{def:CS-matrix}
Let $h : [n] \rightarrow [s]$ be a hash function, and $r:[n] \rightarrow \{-1, 1\}$ be a random sign function. A CS-matrix $\Psi(h, r)\in \{0, 1\}^{s\times n}$ is defined as
\begin{equation*}
\Psi(h, r)_{i,j} = \left\{
  \begin{array}{ll}
    r(j)  & h(j) = i\\
    0  & h(j) \neq i.
  \end{array}
\right.
\end{equation*}   
\end{definition}
Similarly, for a vector $\bx \in \mathbb{R}^n$, we can recover each coordinate of $\bx$ with a bounded error from $\Theta(\log n)$ sketching vectors $\Psi(h, r) \bx$.

\begin{theorem}[\cite{CCFC02}]
\label{thm:count-sketch}
Set $s = \Theta(k/\alpha)$  for an $\alpha \in (0,1)$ and $d = \Theta(\log n)$.  Let $h^1, \ldots, h^d : [n] \to [s]$ be $d$ independent random hash functions, let  $r^1, \ldots, r^d : [n] \to \{-1, 1\}$ be $d$ independent random sign functions, and let $\Psi(h^1, r^1), \ldots, \Psi(h^d, r^d)$ be the corresponding CS-matrices.  Let $\hat{\bx} = (\hat{x}_1, \ldots, \hat{x}_n)$ be a vector such that
$$
\hat{x}_j = \median_{i\in[d]} \left\{ r^i(j) \cdot \left( \Psi(h^i, r^i)\bx\right)_{h^i(j)} \right\}.
$$
We have 
$
\Pr\left [\infnorm{\hat{\bx} - \bx} \leq \alpha / \sqrt{k} \cdot \err{2}(\bx)\right] \geq 1 - 1/n.
$
\end{theorem}


We will use the following sampling matrix.
\begin{definition}[Sampling Matrix]
  \label{def:sample-matrix}
  Let $\Upsilon\in\{0, 1\}^{t \times n}$ be a 0/1 matrix by independently setting for each of the $t$ rows exactly one random coordinate to be $1$. 
\end{definition}



%

\section{Bias-Aware Sketches}
\label{sec:bias}

In this section we propose two efficient bias-aware sketches achieving $\ell_\infty/\ell_1$-guarantee and $\ell_\infty/\ell_2$-guarantee respectively.  

\subsection{Warm Up}
\label{sec:warm-up}

The core of our algorithms is to estimate the bias of the input data.
Before presenting our algorithms, we first discuss a few natural approaches that do {\em not} work, and then illustrate high level ideas of our algorithms.
\smallskip

\noindent{\bf Using mean as the bias.}  The first idea is to use the mean of the input vector $\bx$.  However, this cannot lead to any theoretical error guarantee.  Consider the vector $\bx = (\infty, \infty, 50, 50, 50, 50, 50, 50, 50)$ where $\infty$ denotes a very large number, and $k$ is set to be $2$.  The mean of the coordinates of $\bx$ is $\infty$, but the best bias value is $\beta = 50$ which leads to a tail error $0$ (RHS of (\ref{eq:a-2})).  
Nevertheless, using the mean as the bias may work well in datasets where there are not many extreme values.  We will show in our experiments (Section~\ref{sec:exp}) that this is indeed the case for some real-world datasets.
\smallskip

\noindent{\bf Searching the bias in a post-processing step.}
Another idea is to search the best bias value $\beta$ in a post-processing step after performing the existing sketching algorithms such as Count-Sketch and Count-Median, and then subtract it from the original sketch for the recovery.  More precisely, we can binary search the best $\beta$ by computing the RHS of  (\ref{eq:a-2}) a logarithmic number of times and then picking the best $\beta$ value that minimize the error $\err{p}(\bx - \beta^{(n)})$.
This idea looks attractive since we can just reuse the existing sketching algorithms. 
However, such a post-processing does not fit the streaming setting where we want to answer queries in real-time.  Indeed, in the streaming model we have to redo the binary search of $\beta$ for queries coming in different time steps in the streaming process, which makes the individual point query very slow. 
\smallskip

\noindent{\bf Our approaches.}
In this paper we propose two simple, yet efficient, algorithms to achieve the error guarantee in  (\ref{eq:a-2}), for $p = 1$ and $p = 2$ respectively.  Our algorithms do not need a post-processing step and can thus answer real-time queries in the streaming model.
For $p = 1$, we compute by sampling an approximate median (denoted by $med$) of coordinates in $\bx$, and use it as the bias.  Using the stability of median we can show that $med$ is also an approximate median of the vector $\bx^*$ obtained from $\bx$ by dropping the $k$ ``outliers''.  For  $p = 2$, the idea is still to use the mean.  However, as we have discussed previously, directly using the mean of all items will not give the desired theoretical guarantee, since the mean can be ``contaminated'' by the outliers (extreme values).  We thus choose to employ a Count-Median sketch and use the mean of the ``middle'' buckets in the Count-Median sketch as the bias.   Both algorithms are conceptually very simple, but the complete analysis turns out to be quite non-trivial.  The next two subsections detail our algorithms.

\subsection{Recovery with $\ell_\infty/\ell_1$-Guarantee}
\label{sec:l1}

In this section we give a bias-aware sketch with $\ell_\infty/\ell_1$-guarantee. That is, we try to design a sketching matrix $\Phi \in \mathbb{R}^{t \times n}\ (t \ll n)$ such that from $\Phi \bx$ we can recover an $\hat{\bx}$ satisfying
$
  \|\hat{\bx} - \bx\|_{\infty} = O(1/k)\cdot \min_{\beta}\err{1}(\bx - \beta).
$

\subsubsection{Algorithms}
We use \oneSR\ ($\ell_1$-Sketch/Recover) to denote our algorithm.
Its sketching and recovery procedures are described in Algorithm~\ref{alg:l1-sketch} and Algorithm~\ref{alg:l1-recover} respectively.  For simplicity we assume that the two algorithms can jointly sample hash functions $h^1, \ldots, h^d$ for free (i.e., without any costs).  Indeed, we can simply choose $2$-wise independent hash functions $g, h^i, r^i (i \in [d])$, each of which can be stored in $O(1)$ space. This will not affect any of our mathematical analysis since we will only need to use the second moment of random variables. 
Thus the total extra space to store random hash functions can be bounded by $O(d) = O(\log n)$, and is negligible compared with the sketch size $O(k \log n)$.  In the distributed model we can ask the coordinator to generate these hash functions and then send to all sites, and in the streaming model we can precompute them at the beginning and store them in the memory.

In the sketching phase of \oneSR, we simply use sampling to estimate the best $\beta$ that minimizes $\err{1}(\bx - \beta)$. More precisely, we sample $\Theta(\log n)$ coordinates from $\bx$ and take the median (denoted by $\hat{\beta}$), which we will show is good for the $\ell_\infty/\ell_1$-guarantee. The final (implicit) sketching matrix $\Phi$ is a vertical concatenation of $d = \Theta(\log n)$ independent CM-matrix $\Pi(h^i)$'s and the sampling matrix $\Upsilon$.

In the recovery phase, we use Count-Median to recover $\hat{\bz}$ as an approximation to the {\em de-biased} vector $\bx - \hat{\beta}$; consequently $\hat{\bz} + \hat{\beta}$ will be a good approximation to $\bx$. 

\begin{algorithm}[t]
\DontPrintSemicolon 
\KwIn{$\bx = (x_1, \ldots, x_n) \in \mathbb{R}^n$}
\KwOut{sketch of $\bx$ and a set $S \subseteq \{x_1, \ldots, x_n\}$}
\tcc{assume $s = c_s k$ for a constant $c_s \ge 4$; $d = \Theta(\log n)$; $h^1, \ldots, h^d : [n] \to [s]$ are common knowledge}
generate a sampling matrix $\Upsilon\in\{0,1\}^{20 \log n \times n}$\;
$\forall i \in [d], \ \by^i \leftarrow \Pi(h^i)\bx$\;
$S \leftarrow \Upsilon\bx$\;\label{line:sample-t}
\Return{$S, \{\by^1, \ldots, \by^d\}$}\;
\caption{{\sc $\ell_1$-Sketch$(\bx)$}}
\label{alg:l1-sketch}
\end{algorithm}

\begin{algorithm}[t]
\DontPrintSemicolon 
\KwIn{$S$: a set of randomly sampled coordinates of $\bx$; $\{\by^i =\Pi(h^i)\bx ~|~ i\in[d]\}$}
\KwOut{$\hat{\bx}$ as an approximation of $\bx$}
\tcc{assume $s = c_s k$ for a constant $c_s \ge 4$; $d = \Theta(\log n)$; $h^1, \ldots, h^d : [n] \to [s]$ are common knowledge}
$\hat{\beta} \leftarrow \median \text{ of coordinates in } {S}$\;\label{line:gamma-median}
$\forall i \in [d], \ \boldsymbol{\pi}^i \leftarrow$ coordinate-wise sum of columns of $\Pi(h^i)$\;
$\forall i \in [d], \ \tilde{\by}^i \leftarrow \by^i - \hat{\beta}\boldsymbol{\pi}^i$\; 
\tcc{Run Count-Median recovery}
$\forall j \in [n], \ \hat{z}_j \leftarrow \text{median}_{i\in[d]} \left\{ \left( \tilde{\by}^i \right )_{h^i(j)} \right\}$ \label{eq:e-1}\;
$\hat{\bx} \leftarrow \hat{\bz} + \hat{\beta}$ \label{eq:e-2}\;
\Return{$\hat{\bx}$}\;
\caption{{\sc $\ell_1$-Recover$(S, \{\by^1, \ldots, \by^d\}\})$}}
\label{alg:l1-recover}
\end{algorithm}

\smallskip
The following theorem summarizes the performance of \oneSR.  One can compare it with Theorem~\ref{thm:count-median} for Count-Median.
\begin{theorem}
\label{thm:l1}
There exists a bias-aware sketching scheme such that for any $\bx \in \mathbb{R}^n$, it computes the sketch $\Phi x$, and then recovers an $\hat{\bx}$ as an approximation to ${\bx}$ from $\Phi\bx$ satisfying the following.
\begin{equation}
\label{eq:l1-correctness}
  \Pr[ \|\hat{\bx} - \bx\|_{\infty} \leq C_1/k \cdot \min_{\beta}\err{1}(\bx - \beta) ] \ge 1 - C_2/n,
\end{equation}
where $C_1, C_2 > 0$ are two universal constants.
The sketch can be constructed in time $O(n \log n)$;  the sketch size is bounded by $O(k \log n)$;  the recovery can be done in time $O(n \log n)$.
\end{theorem}

As mentioned in the introduction,  we can convert  $\ell_\infty/\ell_1$ guarantee to $\ell_1/\ell_1$ guarantee.  
\begin{corollary}
\label{cor:l1}
The $\hat{\bx}$ recovered in Theorem~\ref{thm:l1} also guarantees that with probability $1 - O(1/n)$, we have 
$$
\|\hat{\bx} - \bx\|_1 = O(1) \cdot \min_{\beta}\err{1}(\bx - \beta).
$$
\end{corollary}

\subsubsection{Analysis}

\paragraph{Correctness}
\ifdefined\FULL
\else
Due to the space constraints, some proofs will be omitted and they can be found in the full version of this paper~\cite{CZ16e}.
\fi

Let $\bar{\beta}$ be any $\beta$ that minimizes the $\ell_1$-norm error $\err{1}(\bx - \beta)$.  Let $\bx^*$ be the vector obtained by dropping the $k$ coordinates from $\bx$  that deviate the most from $\bar{\beta}$.
We first show:


\begin{lemma}
\label{lem:beta-median}
Given $\bx\in\mathbb{R}^n$, pick any $\bar{\beta} \in \argmin_{\beta}\err{1}(\bx - \beta)$. Let $\bx^*\in \S_{n-k}(\bx)$ be the vector obtained by dropping the $k$ coordinates that deviate the most from $\bar{\beta}$, we must have
\begin{equation}
\label{eq:ell1-1}
\|\bx^* - \bar{\beta}\|_1 = \| \bx^* -\median(\bx^*)\|_1.
\end{equation}
\end{lemma}
\ifdefined\FULL
\begin{proof}
For convenience we assume that $(n-k)$ is odd, and then $\|\bx^* -\beta\|_1$ reaches the minimum only when $\beta = \median(\bx^*)$. It is easy to verify that our lemma also holds when $(n-k)$ is even. Under this assumption, we only need to show $\bar{\beta} = \median(\bx^*)$. 

We prove by contradiction. Suppose $\bar{\beta} \neq \median(\bx^*)$, then 
$$\err{1}(\bx - \median(\bx^*)) \leq \|\bx^* - \median(\bx^*)\|_1 < \err{1}(\bx - \bar{\beta}),$$ contradicting the definition of $\bar{\beta}$.
\end{proof}
\fi

Lemma~\ref{lem:beta-median} gives a more intuitive understanding of the best $\beta$ that minimizes $\err{1}(\bx - \beta)$, and it connects to the idea that the  median of coordinates works. But we are not quite there yet since in (\ref{eq:ell1-1}) we need the {\em exact} median of a vector $\bx^*$ that we do {\em not} know before figuring out $\bar{\beta}$.  To handle this we need the followings two lemmas.  

The first lemma says that a value that is {\em close} (but not necessary equal) to the median of coordinates of $\bx^*$ can be used to approximate the best $\bar{\beta}$.
\begin{lemma}
  \label{lem:median-approx}
  Given a vector $\bx \in \mathbb{R}^m$ with its coordinates sorted non-decreasingly: $x_1 \leq x_2 \leq \ldots \leq x_m$, for any $j$ such that $\frac{m}{4} < j < \frac{3m}{4}$, we have
$$ \sum_{i\in[m]}|x_i - x_j| \leq 2\cdot \min_{\beta}\sum_{i\in[m]}|x_i - \beta| .$$
\end{lemma}
\ifdefined\FULL
\begin{proof}
  For simplicity we assume $m$ is odd; the even case can be handled similarly. Let $t = (m + 1)/2$ be the index of the median coordinate. 
If $j = t$ then we are done. Otherwise, w.l.o.g., we assume $j < t$. We have
\begin{eqnarray*}
  &&\left(\sum_{i=1}^m|x_i - x_j|\right) - \left(\sum_{i=1}^m|x_i - x_t|\right)\\
                                       &=&  \sum_{i=1}^{t-j}i\cdot (x_{t-i+1} - x_{t-i})\\
                                       & =& - (t-j)\cdot x_j + \sum_{i=j+1}^t x_i \\
                                       & =& \sum_{i=j+1}^{t} (x_i - x_j)\\
                                       & \leq& \sum_{i=1}^{m/4} (x_t - x_i) \ \  \left(\text{since}\  \frac{m}{4} < j < t = \frac{m+1}{2}\right) \\
                                       &\leq& \sum_{i=1}^m |x_i - x_t|.
\end{eqnarray*}
The lemma follows.
\end{proof}

\fi

The second lemma says that the median of $O(\log n)$ randomly sampled coordinates of $\bx$ is close to the median of coordinates of the unknown vector $\bx^*$.
\begin{lemma}
  \label{lem:coordinates-sampling}
  Given a vector $\bx \in \mathbb{R}^n$ with its coordinates sorted non-decreasingly: $x_1\leq x_2 \leq \ldots \leq x_n$, if we randomly sample with replacement $t = 20 \log n$ coordinates from $\bx$, then with probability at least $1 - 1/n$ the median of the $t$ samples falls into the range $[x_{n/2-n/6}, x_{n/2+n/6}]$.
\end{lemma}
\ifdefined\FULL
\begin{proof}
Let $X_1, \ldots, X_t$ be the samples we pick. 
The median of them does not fall into the range
$$[x_{n/2-n/6}, x_{n/2+n/6}]$$ if and only if one of the following events happens,
\begin{itemize}
\item $\mathcal{E}_1$: at least half of the samples larger than $x_{n/2+n/6}.$
\item $\mathcal{E}_2$: at least half of the samples smaller than $x_{n/2-n/6}.$
\end{itemize}

We first bound the probability that $\mathcal{E}_1$ happens.
Let $Y_i$ be the random variables such that $Y_i = 1$ if $X_i > x_{n/2+n/6}$, and $Y_i = 0$ otherwise. We have $\E\left[\sum_{i=1}^t Y_i\right] = t/3.$
By a Chernoff bound, we have
\begin{eqnarray*}
&&  \Pr\left[{\sum_{i=1}^tY_i - \frac{t}{3}} > \frac{t}{6}\right]\\
           &\leq& \exp\left(-\frac{t}{12}\right)\\
           &<& 1/(2n).  \quad \quad (t = 20 \log n)
\end{eqnarray*}
Similarly we can show that the probability that $\mathcal{E}_2$ happens is at most $1/(2n)$.
The lemma follows.
\end{proof}
\fi

Now we are ready to prove the theorem.
\begin{proof}(of Theorem~\ref{thm:l1})
W.l.o.g.\ we assume the coordinates of $\bx$ are sorted as $x_1 \leq x_2 \leq \ldots \leq x_n$. To simplify the discussion, we assume $t$ at Line \ref{line:sample-t} of Algorithm \ref{alg:l1-sketch} is odd. The even case can be verified similarly.

Let $\hat{\beta}$ be the median of the $t$ samples in $S$ (Line~\ref{line:gamma-median} in Algorithm~\ref{alg:l1-recover}).  Let $\alpha \in  \argmin_{\beta}\err{1}(\bx - \beta)$. Let $\bx^*$ be the vector obtained by dropping the $k$ coordinates from $\bx$ that deviate the most from $\alpha$. 

By Lemma \ref{lem:coordinates-sampling}, $\hat{\beta} \in [x_{n/2-n/6}, x_{n/2+n/6}]$ holds with probability $1 - 1/n$. Note that we can assume that $k  = O(n / \log n)$ (otherwise the sketch can just be $\bx$ itself which has size $O(k \log n)$). We thus have 
\begin{equation}
\label{eq:c-1}
\Pr\left[(\bx^*)_\frac{(n-k)}{4} \leq \hat{\beta} \leq (\bx^*)_{\frac{3(n-k)}{4}}\right] > 1 - \frac{1}{n}.
\end{equation}
  Applying Lemma \ref{lem:median-approx} to $\bx^*$ (with $m = n-k$), with probability at least $(1 - 1/n)$ it holds that
\begin{align}
  \err{1}(\bx-\hat{\beta}) &\leq \|\bx^* - \hat{\beta}\|_1 \nonumber \\
                   &\leq 2 \cdot \min_{\beta}\|\bx^* - \beta \|_1 \quad (\text{by (\ref{eq:c-1}) and Lemma \ref{lem:median-approx}}) \nonumber \\
                   &= 2\cdot \|\bx^* - \median(\bx^*)\|_1 \nonumber \\
                   &= 2\cdot \|\bx^* - \alpha \|_1 \quad\quad\quad\quad (\text{by  Lemma \ref{lem:beta-median}}) \nonumber \\
                   &= 2\cdot \min_{\beta}\err{1}(\bx - \beta), \label{eq:lemma-6}
\end{align}
where the last equality holds due to the definitions of $\alpha$ and $\bx^*$.
By Theorem~\ref{thm:count-median} (property of Count-Median) and Line \ref{eq:e-1} of Algorithm~\ref{alg:l1-recover} we have
\begin{equation*}
\Pr\left [\infnorm{\hat{\bz} - (\bx - \hat{\beta})} = O\left(\frac{1}{k}\right) \cdot \err{1}(\bx - \hat{\beta})\right] \geq 1 - \frac{1}{n}.
\end{equation*}
Since at Line~\ref{eq:e-2} we set $\hat{\bx} = \hat{\bz} + \hat{\beta}$, we have
\begin{equation}
\label{eq:count-median}
\Pr\left [\infnorm{\hat{\bx} - \bx)} = O\left(\frac{1}{k}\right) \cdot \err{1}(\bx - \hat{\beta})\right] \geq 1 - \frac{1}{n}.
\end{equation}
Inequality (\ref{eq:l1-correctness}) of Theorem~\ref{thm:l1} follows from (\ref{eq:lemma-6})  and (\ref{eq:count-median}).
\end{proof}

\paragraph{Complexities}
Since CM-matrix only has one non-zero entry in each column, using sparse matrix representation we can compute $\Pi(h^i)\bx\ (i \in [d])$ in $O(n)$ time. Thus the sketching phase can be done in time $O(n d) = O(n \log n)$.

The sketch size is $O(k \log n)$ since each $\Psi(h^i)\bx\ (i \in [d])$ has size $O(k)$.

In the recovery phase, the dominating cost is the computation of coordinates in $\hat{\bz}$, for each of which we need $O(d) = O(\log n)$ time.  Thus the total cost is $O(n \log n)$.

\subsection{Recovery with $\ell_\infty/\ell_2$-Guarantee}
\label{sec:l2}

In this section we give a bias-aware sketch with $\ell_\infty/\ell_2$-guarantee. That is, we try to design a sketching matrix $\Phi \in \mathbb{R}^{t \times n}\ (t \ll n)$ such that from $\Phi \bx$ we can recover an $\hat{\bx}$ satisfying
$
  \|\hat{\bx} - \bx\|_{\infty} = O(1/\sqrt{k})\cdot \min_{\beta}\err{2}(\bx - \beta).
$

\subsubsection{Algorithms}
We use \twoSR\ ($\ell_2$-Sketch/Recover) to denote our algorithm.
Its sketching and recovery procedures are described in Algorithm~\ref{alg:l2-sketch} and Algorithm~\ref{alg:l2-recover} respectively.

We again assume that the sketching algorithm and the recovery algorithm can jointly sample (1) independent random hash functions $g, h^1, \ldots, h^d : [n]\rightarrow[s]$ and (2) independent random signed functions $r^1, \ldots, r^d : [n]\rightarrow\{-1, 1\}$ without any costs.  

\begin{algorithm}[t]
\DontPrintSemicolon 
\KwIn{$\bx \in \mathbb{R}^n$}
\KwOut{the sketch of $\bx$}
\tcc{assume $s = c_s k$ for a constant $c_s \ge 4$; $d = \Theta(\log n)$; $g, h^1, \ldots, h^d : [n] \to [s]$; $r^1, \ldots, r^d : [n] \to \{-1, 1\}$ are common knowledge}
$\bw \leftarrow \Pi(g)\bx$\;
$\forall i \in [d], \ \by^i \leftarrow \Psi(h^i, r^i)\bx$\;

\Return{$\bw, \{\by^1, \ldots, \by^d\}$}
\caption{{\sc $\ell_2$-Sketch$(\bx)$}}
\label{alg:l2-sketch}
\end{algorithm}

\begin{algorithm}[t]
\DontPrintSemicolon 
\KwIn{$\bw=\Pi(g)\bx$; $\{\by^i =\Psi(h^i, r^i)\bx ~|~ i\in[d]\}$}
\KwOut{$\hat{\bx}$ as an approximation of $\bx$}
\tcc{assume $s = c_s k$ for a constant $c_s \ge 4$; $d = \Theta(\log n)$; $g, h^1, \ldots, h^d : [n] \to [s]$; $r^1, \ldots, r^d : [n] \to \{-1, 1\}$ are common knowledge}
$\boldsymbol{\pi} \leftarrow$ coordinate-wise sum of columns of $\Pi(g)$ \label{line:pi} \;
w.l.o.g.\ assume $w_1/\pi_1 \le \ldots \le w_s/\pi_s$; set
$\hat{\beta} = \sum_{i = s/2-k}^{s/2+k-1}w_i \left/\sum_{i = s/2-k}^{s/2+k-1}\pi_i \right.$\;\label{line:beta} 
$\forall i \in [d], \ \boldsymbol{\psi}^i \leftarrow$ coordinate-wise sum of columns of $\Psi(h^i, r^i)$\; \label{line:psi}
$\forall i \in [d], \ \tilde{\by}^i \leftarrow \by^i - \hat{\beta}\boldsymbol{\psi}^i$\; \label{line:debias}
\tcc{Run the Count-Sketch recovery}
$\forall j \in [n], \ \hat{z}_j \leftarrow \text{median}_{i\in[d]} \left\{ r^i(j)\cdot \left( \tilde{\by}^i \right )_{h^i(j)} \right\}$  \;\label{line:median}
$\hat{\bx} \leftarrow \hat{\bz} + \hat{\beta}$\; \label{line:change}
\Return{$\hat{\bx}$}\;

\caption{{\sc $\ell_2$-Recover$(\bw, \{\by^1, \ldots, \by^d\})$}}
\label{alg:l2-recover}
\end{algorithm}

In our algorithms we first use the CM-matrix to obtain a good approximation $\hat{\beta}$ of the $\beta$ that minimizes $\err{2}(\bx - \beta)$, and then use the Count-Sketch algorithm to recover $\hat{\bz}$ as an approximation to the {\em de-biased} vector $\bx - \hat{\beta}$; and consequently $\hat{\bz} + \hat{\beta}$ will be a good approximation to $\bx$.  The final (implicit) sketching matrix $\Phi \in \mathbb{R}^{s(d+1) \times n}$  in Algorithm~\ref{alg:l2-sketch} is a vertical concatenation of a CM-matrix $\Pi(g)$ and $d = \Theta(\log n)$ independent CS-matrices $\Psi(h^i, r^i)$'s.

In Algorithm~\ref{alg:l2-recover}, to approximate the best $\beta$ we first sum up all the columns of $\Pi(g)$, giving a vector $\bpi = (\pi_1, \ldots, \pi_s)$ (Line~\ref{line:pi}).  Let $\bw = \Pi(g) \bx \in \mathbb{R}^s$.  W.l.o.g.\ assume that $w_1/\pi_1 \le \ldots \le w_s/\pi_s$. We estimate $\beta$ by 
$$\textstyle \hat{\beta} = \sum_{i = s/2-k}^{s/2+k-1}w_i \left/ \sum_{i = s/2-k}^{s/2+k-1}\pi_i \right. .$$ 

The intuition of this estimation is the following. First note that $w_i/\pi_i\ (i \in [s])$ is the average of coordinates of $\bx$ that are hashed into the $i$-th coordinate(bucket) of sketching vector $\Pi(g)\bx$. In the case that there is no ``outlier'' coordinate of $\bx$ that is hashed into the $i$-th bucket of $\Pi(g)\bx$, then $w_i/\pi_i\ (i \in [s])$ should be close to the best bias $\beta$.  Since there are at most $k$ outliers, if we choose $s \ge 4k$ then most of these $s$ buckets in $\Pi(g)$ will {\em not} be ``contaminated'' by outliers.  

The next idea is to sort the buckets according to the average of coordinates of $\bx$ hashed into it ({\em i.e.}, $w_i/\pi_i$), and then choose the $2k$ buckets around the median and take the average of coordinates hashed into those buckets (Line~\ref{line:beta}). We can show that the average of coordinates of $\bx$ that are hashed into these $2k$ ``median'' buckets is a good estimation of the best $\beta$.  Note that there could still be outliers hashed into the median $2k$ buckets, but we are able to prove that such outliers will not affect the estimation of $\beta$ by much. After getting an estimate of $\beta$ we de-bias the sketching vector $\by$ (Line~\ref{line:psi} and \ref{line:debias}) for the next step recovery (Line~\ref{line:median} and \ref{line:change}).

\smallskip

The following theorem summarizes the performance of \twoSR.    One can compare it with Theorem~\ref{thm:count-sketch} for Count-Sketch.
\begin{theorem}
\label{thm:l2}
There exists a bias-aware sketching scheme such that for any $\bx \in \mathbb{R}^n$, it computes $\Phi \bx$, and then recovers an $\hat{\bx}$ as an approximation to $\bx$ from $\Phi \bx$ satisfying the following:
\begin{equation}
\label{eq:l2-correctness}
 \Pr[ \|\hat{\bx} - \bx\|_{\infty} \leq C_1 /\sqrt{k}\cdot \min_{\beta}\err{2}(\bx - \beta) ] \ge 1 - C_2 /n,
\end{equation}
where $C_1, C_2 > 0$ are two universal constants.
The sketch can be constructed in time $O(n \log n)$;  the sketch size is bounded by $O(k \log n)$; the recovery can be done in time $O(n \log n)$.
\end{theorem}

As mentioned in the introduction,  we can convert $\ell_\infty/\ell_2$ guarantee to $\ell_2/\ell_2$ guarantee. 

\begin{corollary}
\label{cor:l2}
The $\hat{\bx}$ recovered in Theorem~\ref{thm:l2} also guarantees that with probability $1 - O(1/n)$, we have
$$
\|\hat{\bx} - \bx\|_2 = O(1) \cdot \min_{\beta}\err{2}(\bx - \beta).
$$
\end{corollary}

\subsubsection{Analysis}
\label{sec:l2-analysis}

\paragraph{Correctness}
\ifdefined\FULL
\else
Again due to the space constraints, some proofs will be omitted and they can be found in the full version of this paper~\cite{CZ16e}.
\fi

Similar to the $\ell_1$ case, we first replace the somewhat obscure expression $\min_{\beta}\err{2}(\bx - \beta)$ in Theorem~\ref{thm:l2} with another one which is more convenient to use.  


\begin{lemma}
\label{lem:min-k}
For any $\bx \in \mathbb{R}^n$ and $k < n$, let $\bx^*$ be a vector in $\S_{n-k}(\bx)$ that has the minimum variance.  It holds that
\begin{eqnarray}
  \label{eq:beta-x}
  \left(\min_{\beta} \err{2}(\bx - \beta)\right)^2 &=& (n-k)  \sigma^2(\bx^*) \nonumber\\
  &=& \|\bx^* - \mean(\bx^*)\|_2^2.
\end{eqnarray}
Furthermore,  $\bx^*$ is equivalent to the vector obtained by dropping the $k$ coordinates from $\bx$ that deviate the most from $\mean(\bx^*)$.
\end{lemma}
\ifdefined\FULL
\begin{proof}
First, by the definition of $\bx^*$ we have
$$(n-k) \sigma^2(\bx^*) =  \min_{\bx'\in \S_{n-k}(\bx)} \|\bx' - \mean(\bx')\|_2^2.$$ 
By the definition of $\err{2}(\cdot)$, we have
\begin{align*}
\min_{\beta} \err{2}(\bx - \beta) &\geq  \min_{\bx'\in \S_{n-k}(\bx)} \min_{\beta}\|\bx' - \beta\| \\
  &= \min_{\bx'\in \S_{n-k}(\bx)} \|\bx' - \mean(\bx')\|_2.  
\end{align*}

Thus to prove $(\ref{eq:beta-x})$, it suffices to show that
$$\min_{\beta} \err{2}(\bx - \beta) \leq \min_{\bx'\in \S_{n-k}(\bx)} \|\bx' - \mean(\bx')\|_2.$$

Since the order of the coordinates in $\bx$ do not matter, w.l.o.g.\ we assume $\bx^* = (x_1, x_2, \ldots, x_{n-k})$.  Let $\gamma = \mean(\bx^*)$, and write $x_i = \gamma + \Delta_i$, or equivalently, $\Delta_i = x_i - \gamma$.  Note that
if 
\begin{equation}
  \label{eq:assumption}
  \min_{i\in [n]\backslash[n-k]}|\Delta_i| \geq \max_{i\in [n-k]}|\Delta_i|, 
\end{equation}
then we are done because
\begin{align}
  \min_{\beta} \err{2}(\bx - \beta)^2 &\leq \err{2}(\bx - \gamma)^2 \nonumber\\
           &=  \sum_{i\in[n-k]} \Delta_i^2 =  \|\bx^* - \gamma\|_2^2.    \label{eq:implication} 
\end{align}
Now we assume $(\ref{eq:assumption})$ is false.  Again w.l.o.g., we assume $|\Delta_1| = \max\limits_{i\in [n-k]}|\Delta_i|$ and $|\Delta_n| = \min\limits_{i\in [n]\backslash[n-k]}|\Delta_i|$, then $|\Delta_1| > |\Delta_n|$. Let $\bx' = (x_2, x_3, \ldots, x_{n-k-1}, x_{n-k}, x_n) \in \S_{n-k}$, that is, $\bx'$ is obtained by dropping $x_i$ from $\bx^*$ and then appending $x_n$, we have 
\begin{align*}
\|\bx' - \mean(\bx')\|_2^2 &\leq \|\bx' - \mean(\bx^*)\|_2^2 \\
& =  \|\bx' - \gamma\|_2^2 \quad \quad \text{(by definition of $\gamma$)}\\
                          & =  \Delta_n^2 + \sum_{i=2}^{n-k} \Delta_i^2 \quad \text{(by definition of $\Delta_i$)}\\
                          & < \Delta_1^2 + \sum_{i=2}^{n-k}\Delta_i^2 \\
                          & =  \|\bx^* - \mean(\bx^*) \|_2,
\end{align*}
which contradicts the definition of $\bx^*$. Hence $(\ref{eq:assumption})$ holds, and consequently $(\ref{eq:beta-x})$ holds.  

On the other hand, $(\ref{eq:assumption})$ also implies that $x_{n-k+1},  \ldots, x_n$ are the $k$ coordinates of $\bx$ that deviate the most from $\gamma=\mean(\bx^*)$.  
\end{proof}
\fi

We then show (using Lemma~\ref{lem:min-k}) that a good approximation of $\mean(\bx^*)$ is also a good approximation of the best $\beta$.
\begin{lemma}
  \label{lem:deviation}
  For any $\bx\in\mathbb{R}^n$ and $k < n$, let $\bx^*$  be a vector in $\S_{n-k}(\bx)$ that has the minimum variance. For any $\alpha$ such that $\abs{\mean(\bx^*) - \alpha}^2 \leq C \cdot \sigma^2(\bx^*)$ for any constant $C > 0$, we have
  $$\err{2}(\bx - \alpha)^2 = O\left(\min_{\beta}\err{2}(\bx - \beta)^2\right).$$
\end{lemma}
\ifdefined\FULL
\begin{proof}
 W.l.o.g.\ we again assume $\bx^* = (x_1, \ldots, x_{n-k})$.
  Define $f(b) \triangleq \|\bx^*-b\|_2^2$. Let $\gamma = \mean(\bx^*)$. By Lemma \ref{lem:min-k} we have 
\begin{equation}
\label{eq:b-1}
  f(\gamma) =  (n-k)\sigma^2(\bx^*) = \|\bx^*-\gamma\|_2^2 = \min_{\beta}\err{2}(\bx - \beta)^2.
\end{equation}
  Write $\alpha = \gamma + \Delta$ and thus 
  $\Delta^2 \leq C \sigma^2(\bx^*)$,
  \begin{align*}
    \err{2}(\bx - \alpha)^2  &\leq \|\bx^* - \alpha\|_2^2\\
                             & = f(\alpha)  = f(\gamma + \Delta)\\
                           & = \sum_{i\in[n-k]} \left( ( x_i - \gamma) - \Delta \right)^2 \\
                           & = (n-k)\Delta^2 + \sum_{i = 1}^{n-k} (x_i - \gamma)^2 - 2\Delta\sum_{i=1}^{n-k}(x_i - \gamma)\\
                           & \leq   (n-k) \cdot C \sigma^2(\bx^*) + \norm{\bx^* - \gamma}_2^2 + 0\\
                           & = O\left(\min_{\beta}\err{2}(\bx - \beta)^2\right). \quad \quad \text{(by (\ref{eq:b-1}))}
  \end{align*}
We are done.
\end{proof}
\fi

The next lemma is crucial. It shows that the approximation $\hat{\beta}$ obtained in the recovery algorithm (Algorithm~\ref{alg:l2-recover}) is a good approximation of $\mean(\bx^*)$.  
\begin{lemma}
  \label{lem:estimate-beta}
Let $\hat{\beta}$ be given at Line~\ref{line:beta} of Algorithm \ref{alg:l2-recover}.  If $s = c_s k$ for a sufficiently large constant $c_s \ge 4$,  it holds that
$$ \Pr\left[\left(\hat{\beta} - \mean(\bx^*)\right)^2 = O\left(\sigma^2(\bx^*)\right)\right] = 1-O\left(\frac{1}{n}\right)$$
for any $\bx^*$ in $\S_{n-k}(\bx)$ that has the minimum variance.
\end{lemma}

\ifdefined\FULL
Before proving Lemma~\ref{lem:estimate-beta}, we need a bound on the difference between the average of all coordinates of a vector and the average of a subset of coordinates. 
\begin{lemma}
\label{lem:ave-var}
Let $\bx = \{x_1,\ldots, x_m\} \in \mathbb{R}^m$ be a vector.   Let $S$ be a subset of $\bx$'s coordinates of size $\abs{S} = \Theta(m)$.  Let $\mu = \frac{1}{m} \sum_{i \in [m]} x_i$, and $\mu' = \frac{1}{\abs{S}}\sum_{i \in S} x_i$. Then we have 
\begin{equation*}
\abs{\mu' - \mu}^2 = O\left(\sigma^2(\bx)\right).
\end{equation*}
\end{lemma}
\begin{proof}
W.l.o.g., let $x_1  \le \ldots \le x_m$.  Let $\alpha \in (0,1)$ be an arbitrary constant. Let 
$$\mu_1 = \frac{1}{\alpha m}  \sum_{i \in [\alpha m]}x_i \quad \text{and} \quad \mu_2 = \frac{1}{(1-\alpha) m}  \sum_{i \in [m] \backslash [\alpha m]}x_i.$$ We only need to prove two extreme cases where $S = [\alpha m]$ or $S = [m] \backslash [\alpha m]$ (in both cases $\abs{S} = \Theta(n)$ since $\alpha = (0,1)$ is an arbitrary constant):
\begin{equation}
\label{eq:ave-var}
\abs{\mu_1 - \mu}^2 \le \frac{1}{\alpha} \cdot \sigma^2(\bx) \quad \text{and} \quad \abs{\mu_2 - \mu}^2 \le \frac{1}{1-\alpha} \cdot \sigma^2(\bx).
\end{equation}
For simplicity (and w.o.l.g.), we assume $\mu = 0$, since we can always define $y_i = x_i - \mu$ and prove on $y_i$'s.  We can write the variance of $\bx$ as 
\begin{eqnarray*}
\sigma^2(\bx) & = & \frac{1}{m} \left(\sum_{i \in [\alpha m]} x_i^2  +  \sum_{i \in [m] \backslash [\alpha m]} x_i^2 \right) \\
&\ge& \left(\alpha m \mu_1^2 + (1-\alpha) m \mu_2^2\right)/m \quad \text{(Cauchy-Schwarz)} \\
&=& \alpha \mu_1^2 + (1-\alpha) \mu_2^2.
\end{eqnarray*}
(\ref{eq:ave-var}) follows straightforwardly.
\end{proof}

\smallskip

\begin{proof} (of Lemma~\ref{lem:estimate-beta})
Fix any $\bx^* = (x_{i_1}, \ldots, x_{i_{n-k}})$ in $\S_{n-k}(\bx)$ that has the minimum variance.  Let 
$O$ be the set of the top-$k$ indices $i$ in $\bx$ that maximize $\abs{x_i - \mean(\bx^*)}$.  By Lemma~\ref{lem:min-k} we have that $\bx^*$ can be obtained from $\bx$ by dropping coordinates indexed by $O$.

We call an index $i \in [s]$ in the sketching vector $\bw = \Pi(g) \bx$ {\em contaminated} if there is at least one $o \in O$ such that $g(o) = i$. W.l.o.g., we assume $w_1/\pi_1 \le \ldots \le w_s/\pi_s$ where $\bpi$ is defined at Line~\ref{line:pi} of Algorithm~\ref{alg:l2-recover}.  Let 
$${I} = \{i\ |\ s/2 - k \le i < s/2 +k\}$$ be the $2k$ ``median'' indices of $\bw$, and $$\bar{{I}} = \{i\ |\ i < s/2 - k \vee i \ge s/2 +k\}$$ be the rest of indices in $\bw$.  Since $\abs{O} = k$ and $s \ge 4k$, there are {\em at most} $k$ coordinates in $I$ that are contaminated, and {\em at least} $k$ coordinates in $\bar{I}$ that are {\em not} contaminated.  

The approximation to the bias $\beta$ at Line~\ref{line:beta} of Algorithm~\ref{alg:l2-recover} can be written as
\begin{equation}
\label{eq:beta}
\hat{\beta} = \frac{\sum_{i \in {I}} w_i}{\sum_{i \in {I}} \pi_i}.
\end{equation}
Let $O' = I \cap g(O)$ be the indices in $I$ that are contaminated, and let $J$ be an arbitrary subset of $\bar{I}$ with size $\abs{O'}$.  Define 
\begin{equation}
\label{eq:gamma}
\gamma_1 = \min_J \frac{\sum_{i \in {I \cup J \backslash O'}} w_i}{\sum_{i \in {I \cup J \backslash O'}} \pi_i}  \ \ \text{and} \ \ \gamma_2 = \max_J \frac{\sum_{i \in {I \cup J \backslash O'}} w_i}{\sum_{i \in {I \cup J \backslash O'}} \pi_i}.
\end{equation}
It is easy to see that $\gamma_1 \le \hat{\beta} \le \gamma_2$: since $s \ge 4k$, one can always find a subset $J \subseteq \bar{I}$ of size $\abs{O'}$ such that for any $j \in J, o \in O'$ we have $w_j / \pi_j \ge w_o / \pi_o$, and replacing $O'$ with $J$ only increases the RHS of (\ref{eq:beta}); On the other hand one can also find a subset $J \subseteq \bar{I}$ of size $\abs{O'}$ such that for any $j \in J, o \in O'$ we have $w_j / \pi_j \le w_o / \pi_o$, and replacing $O'$ with $J$ only decreases  the RHS of (\ref{eq:beta}).

We now show that both $\gamma_1$ and $\gamma_2$ deviate at most $O(\sigma(\bx^*))$ from $\mean(\bx^*)$, and consequently $\hat{\beta}$, which is sandwiched by $\gamma_1$ and $\gamma_2$, deviates from $\mean(\bx^*)$ by at most $O(\sigma(\bx^*))$.  Consider the set $G = g^{-1}(I \cup J \backslash O')$. First, by definitions of $I, J$ and $O'$ we have $G \subseteq \{{i_1}, \ldots, {i_{n-k}}\}$; and thus $\{x_j\ |\ j \in G\}$ are coordinates in $\bx^*$. Second, since $\abs{I \cup J \backslash O'} = \Theta(k)$ and $g$ is a random mapping from $[n]$ to $[s]$, by a Chebyshev inequality we have $\abs{G} = \Theta(n)$ with probability at least $1 - O(1/n)$.\footnote{More precisely, define for each $i \in [n]$ a random variable $Y_i$, which is $1$ if $g(i) \in I \cup J \backslash O'$ and $0$ otherwise. Since $\abs{I \cup J \backslash O'} = \Theta(k)$ and $s = \Theta(k)$, we have $\E[Y_i] = \Theta(1)$, and $\var[Y_i] \le \E[Y_i^2] = O(1)$.  Next note that $\abs{G} = \sum_{i \in [n]} Y_i$. We thus can apply a Chebyshev inequality on $Y_i$'s and conclude that $\abs{G} = \Theta(n)$ with probability $1-O(1/n)$.} 
For any $J \subseteq \bar{I}$ of size $\abs{O'}$, let 
$$\gamma_J = \frac{1}{\abs{G}} \sum_{j \in G} x_j = \frac{\sum_{i \in {I \cup J \backslash O'}} w_i}{\sum_{i \in {I \cup J \backslash O'}} \pi_i}.$$
By Lemma~\ref{lem:ave-var}, we have
\begin{equation}
\label{eq:G}
\abs{\gamma_J - \mean(\bx^*)} = O(\sigma(\bx^*)).  
\end{equation} 
Since Inequality~(\ref{eq:G}) applies to any $J \subseteq \bar{I}$ of size $\abs{O'}$, we have $\abs{\gamma - \mean(\bx^*)} = O(\sigma(\bx^*))$ for any $\gamma \in \{\gamma_1, \gamma_2\}$.
\end{proof}
\fi

Finally we prove Theorem~\ref{thm:l2} using Lemma~\ref{lem:deviation} and \ref{lem:estimate-beta}; we show that the obtained $\hat{\beta}$ is a good approximation of the best $\beta$ that minimizes $\err{2}(\bx - \beta)$.

\begin{proof}(of Theorem~\ref{thm:l2})
Let $\bx^*$ be a vector in $\S_{n-k}(\bx)$ that has the minimum variance. At Line \ref{line:debias}-\ref{line:median} in Algorithm \ref{alg:l2-recover} the Count-Sketch recovery algorithm is used to compute $\hat{\bz}$ as an approximation to $\bx - \hat{\beta}$. By Theorem~\ref{thm:count-sketch} we have
\begin{equation*}
  \Pr\left [\infnorm{\hat{\bz} - (\bx - \hat{\beta})} = O\left(\frac{1}{\sqrt{k}}\right)\cdot\err{2}(\bx-\hat{\beta})\right] \geq 1 - \frac{1}{n}.
\end{equation*}
Since at Line~\ref{line:change} we set $\hat{\bx} = \hat{\bz} + \hat{\beta}$, it holds that
\begin{equation}
  \label{eq:count-sketch}
  \Pr\left [\infnorm{\hat{\bx} - \bx} = O\left(\frac{1}{\sqrt{k}}\right)\cdot\err{2}(\bx-\hat{\beta})\right] \geq 1 - \frac{1}{n}.
\end{equation}
By Lemma \ref{lem:estimate-beta}, 
$$  \Pr\left[|\hat{\beta} - \mean{\bx^*}| = O\left(\sigma(\bx^*)\right)\right] = 1 - O\left(\frac{1}{n}\right). $$
Plugging it to Lemma \ref{lem:deviation} we have with probability at least $(1 - O(1/n))$ that
\begin{equation}
  \label{eq:deviation}
  \err{2}(\bx - \hat{\beta}) = O\left(\min_{\beta}\err{2}(\bx - \beta)\right).
\end{equation}
Inequality (\ref{eq:l2-correctness}) in Theorem~\ref{thm:l2} follows from
$(\ref{eq:count-sketch})$ and $(\ref{eq:deviation})$.
\end{proof}

\paragraph{Complexities}  Since each CS-Matrix or CM-matrix only has one non-zero entry in each column, using sparse matrix representation we can compute $\Psi(h^i, r^i)\bx\ (i \in [d])$ or $\Pi(g)\bx$ in $O(n)$ time. Thus the sketching phase can be done in time $O(n d) = O(n \log n)$.

The sketch size is $O(k \log n)$, simply because $\Pi(g)\bx$ and each $\Psi(h^i, r^i)\bx\ (i \in [d])$ has size $O(k)$.

In the recovery phase, the dominating cost is the computation of coordinates in $\hat{\bz}$, for each of which we need $O(d) = O(\log n)$ time.  Thus the total cost is $O(n \log n)$.

\subsection{Streaming Implementations}
\label{sec:streaming}

\ifdefined\FULL

We now discuss how to maintain the bias (estimation) $\beta$ at any time step in the streaming setting. This is useful since we would like to answer individual point queries efficiently without decoding the whole vector $\bx$; to this end we need to first maintain $\beta$ efficiently.  

For the $\ell_\infty/\ell_1$ guarantee we can easily maintain a good approximation of $\beta$ with $O(\log \log n)$ time per update: we can simply keep the $\Theta(\log n)$ sampled coordinates sorted (e.g., using a balanced binary search tree) during the streaming process, and use their median as an approximation of $\beta$ at any time step. For the $\ell_\infty/\ell_2$ guarantee, the recovery procedure in Algorithm~\ref{alg:l2-recover} takes the average of items in the middle $2k$ of the $s$ sorted buckets $w_1/\pi_1 \le \ldots \le w_s/\pi_s$.  This can again be done in $O(\log s) = O(\log k + \log\log n)$ time per update using a balanced binary search tree. An alternative implementation using {\em biased heaps} is described Algorithm~\ref{algo:bias-heap}.  The idea of the algorithm is very simple: we use heaps to keep track of $\sum_{i \in A} w_i, \sum_{i \in C} w_i, \sum_{i \in A} \pi_i, \sum_{i \in C} \pi_i$, where $A$ is the set of the top $(s/2 - k)$ coordinates and $C$ is the set of the bottom $(s/2 - k)$ coordinates (the order is defined by $w_i/\pi_i$). Using these sums together with $\sum_{i \in [s]} w_i$ and  $\norm{\pi}_1$ we can well estimate the bias.  

\begin{algorithm}[!ht]
\DontPrintSemicolon 
\KwIn{$s$ (\# of rows of CM-matrix $\Pi$), and vector $\bpi$ (coordinate-wise sum of columns of matrix $\Pi$)}
create $s$ nodes, each of which is associated with a (key, value, id) triple  $(w_i/\pi_i, w_i, i)$ where $w_i = 0$ and $\pi_i$ is the $i$-th coordinate of $\bpi$ \;
\tcc{key is the priority of nodes in the heap}
set $k \leftarrow s/4$\;
initialize a min-heap $A$ for nodes $1$ to $s/2 - k - 1$\; 
initialize a max-heap $B$ for nodes $s/2 - k$ to $s$\;
initialize a max-heap $C$ for nodes $s/2 + k$ to $s$\;
initialize a min-heap $D$ for nodes $1$ to $s/2 + k - 1$\;
$\pi_A \triangleq \sum_{i \in \{\text{all id in}~ A\}} \pi_i$;
$\pi_C \triangleq \sum_{i \in \{\text{all id in}~ C\}} \pi_i$\;
$w_A \triangleq \sum_{i \in \{\text{all id in}~ A\}} w_i$;
$w_C \triangleq \sum_{i \in \{\text{all id in}~ C\}} w_i$\;
$w \leftarrow 0$\;

\tcc{process updates or queries}
\Case{upon receiving an update $(j, \Delta)$} {
  $w \leftarrow w + \Delta$\;
  find node with id $j$ in two of heaps $A, B, C, D$ and update its $w_j$ by adding $\Delta$; update the corresponding key and maintain the heap properties\;
  \If {the key of the top node in $A$ is smaller than that of the top of $B$}{
    swap their tops and maintain heap properties\;
  }
  \If {the key of the top node of $C$ is larger than that of the top of $D$}{
    swap their tops and maintain heap properties\;
  }
  update $\pi_A, \pi_C, w_A, w_C$ if necessary\;
}
\Case{upon receiving a query of the bias $\beta$} {
  \Return{ $\frac{w - w_A - w_C}{\norm{\bpi}_1 - \pi_A - \pi_C}$ }\;
}
\caption{{Bias-Heap}}
\label{algo:bias-heap}
\end{algorithm}

The full streaming algorithm for $\ell_\infty/\ell_2$ guarantee is described in Algorithm~\ref{algo:l2-streaming}, which is similar to Algorithm~\ref{alg:l2-recover} but has been augmented to fit the streaming model.  The one for $\ell_\infty/\ell_1$ guarantee can be done similarly, and we omit here.

\begin{algorithm}[t]
\DontPrintSemicolon 
\tcc{$s = c_s k$ for a constant $c_s \ge 4$; $d = \Theta(\log n)$; $g, h^1, \ldots, h^d : [n] \to [s]$; $r^1, \ldots, r^d : [n] \to \{-1, 1\}$ are the same as in Algorithm~\ref{alg:l2-recover}; $\bpi \gets \sum_{j \in [n]} \text{$j$-th column of } \Pi(g)$}
$\forall i\in[d], \boldsymbol{\psi}^i \leftarrow$ coordinate-wise sum of columns of the CS-matrix $\Psi(h^i, r^i)$\;
initialize $\by^1, \ldots \by^d$ to be all-zero vectors of length $s$\; 
initialize Bias-Heap in  Algorithm~\ref{algo:bias-heap} with $s$ and $\bpi$\;
\tcc{process updates or queries}
\Case {upon an update $(e_i, \Delta)$} {
  $\forall t \in [d]$, $y^t_{h^t(i)} \leftarrow y^t_{h^t(i)} + r^t(i)\cdot\Delta$\;
  update the Bias-Heap with $(g(e_i), \Delta)$\;
}
\Case {upon receiving a query for computing $x_i$} {
  query Bias-Heap to get $\hat{\beta}$\;
  $z \leftarrow \median{\left\{r^t(i) \cdot \left. \left( y^t_{h^t(i)} - \psi^t_{h^t(i)}\cdot \hat{\beta}\right)~\right|~ t \in [d]\right\} }$\;
  \Return{$z + \hat{\beta}$}\;
}
\caption{Streaming Algorithm with $\ell_\infty/\ell_2$ Guarantee}
\label{algo:l2-streaming}
\end{algorithm}

Finally we comment on how to generate and store random hash functions in the streaming setting.  In fact, we can simply choose hash functions $g, h^i, r^i (i \in [d])$ to be 2-wise independent (and each will use $O(1)$ space to store). This will not affect any of our analysis since we only need to use the second moment of random variables (same in the proofs for Theorem~\ref{thm:count-median} and Theorem~\ref{thm:count-sketch} for the CM-sketch and CS-sketch, see \cite{CCFC02,CM04}).
Thus the total extra space to store random hash functions can be bounded by $O(d) = O(\log n)$, and is negligible compared with the sketch size $O(k \log n)$.

\else

Finally we discuss how to maintain the bias (estimation) $\beta$ at any time step in the streaming setting. This is useful since we would like to answer individual point queries efficiently without decoding the whole vector $\bx$; to this end we need to first maintain $\beta$ efficiently.  

For \oneSR, we can easily maintain a good approximation of $\beta$ with $O(\log \log n)$ time per update: we can simply keep the $\Theta(\log n)$ sampled coordinates sorted during the streaming process, and use their median as an approximation of $\beta$ at any time step. 

For \twoSR, the recovery procedure in Algorithm~\ref{alg:l2-recover} takes the average of items in the middle $2k$ of the $s$ sorted buckets $w_1/\pi_1 \le \ldots \le w_s/\pi_s$.  This can again be done in $O(\log s) = O(\log k + \log\log n)$ time per update using a balanced binary search tree. An alternative implementation is to use {\em biased heap}  which has the same update/query complexity but works faster in practice. Biased heap is  described in Algorithm 5 of the full version of this paper \cite{CZ16e}.  The whole streaming implementation can also be found in Algorithm 6 of \cite{CZ16e}.

\fi

\newcommand{\ellerror}{\textbf{$\ell_1$-{\tt error}}}
\newcommand{\ellmeanOne}{{$\ell_1$-\textsf{mean}}}
\newcommand{\ellmeanTwo}{{$\ell_2$-\textsf{mean}}}

\section{Experiments}
\label{sec:exp}

In this section we give our experimental studies. 

\subsection{The Setup}
\label{sec:setup}

\paragraph{Reference Algorithms}
We compare \oneSR\ and \twoSR\ with Count-Sketch (\CS) algorithm and Count-Median Sketch  (\CM) algorithm, as well as non-linear sketches Count-Min with conservative update (\CMCU) \cite{EV02,GDC12} and Count-Min-Log with conservative update (\CMLCU) \cite{PF15}.  For \CMLCU, we set the base to be $1.00025$. 

We would like to mention that the Count-Min algorithm, which was proposed in the same paper \cite{CM04} as Count-Median, is very similar to Count-Median; they share the same sketching matrix. In fact, Count-Median can be thought as a generalization of Count-Min~\cite{CM04}.  On the other hand, \CMCU\ is an improvement upon Count-Min and has strictly better performance. We thus do not compare our algorithms with Count-Min  in our experiments. 



Finally, we also compare \oneSR\ and \twoSR\ with two simple algorithms that just use the {\em mean} of {\em all} coordinates in $\bx$ as the bias (denoted by \ellmeanOne\ and \ellmeanTwo\ respectively).
\ifdefined\FULL
See sec. \ref{sec:ell-mean} for details.
\else
Due to the space constraints we leave this part to the full version \cite{CZ16e}. 
\fi
 As mentioned earlier, using the mean of all the coordinates as the bias {\em cannot} give us any theoretical guarantees -- for example, it will perform badly on datasets where the top-$k$ largest coordinates are significantly larger than all the rest coordinates. However, this simple heuristic does work well on some real-world datasets.  Thus they may be interesting to practitioners.

\paragraph{Datasets}
We compare the algorithms using a set of real and synthetic datasets. 
\begin{itemize}
\item {\tt Gaussian}. Each entry of $\bx$ is independently sampled from the Gaussian distribution  $\mathcal{N}(b, \sigma^2)$ where $b$ is the {\em bias}. In our experiments, we fix $n=500,000,000$, $\sigma = 15$ and vary the value of $b$. 

\ifdefined\FULL
\item {\tt Gaussian-2}. This dataset is used to compare \oneSR, \twoSR, \ellmeanOne\ and \ellmeanTwo. Each entry of $\bx$ is independently sampled from the Gaussian distribution  $\mathcal{N}(100, 15^2)$. In our experiments, we fix $n=5,000,000$. To verify our theorems, we shift several coordinates. See Figure \ref{fig:gauss-2} for details.
\fi
  

\ifdefined\WORLDCUP
\item{\tt WorldCup} \cite{arlitt1998world}. This dataset consists of all the requests made to all resources of the 1998 World Cup Web site between April 30, 1998 to July 26, 1998.  We picked all the requests made to all resources on May 14, 1998. We construct $\bx$ from those requests where each coordinate is the number of requests made in a particular second. The dimension of $\bx$ is therefore $24\times 3600 = 86,400$. There are about $3,200,000$ requests. 
\fi

\item {\tt Wiki} \cite{wiki}. This dataset contains pageviews to the English-language Wikipedia from March 16, 2015 to April 25, 2015. The number of pageviews of each second is recorded. We model the data as a vector $\bx$ of length about $3,500,000$ (we added up mobile views and desktop views if they have the same timestamp).  There are about $13,000,000,000$ pageviews.

\item {\tt Higgs} \cite{BSW14}. The dataset was produced by Monte Carlo simulations. There are $28$ kinematic properties measured by the particle detectors in the accelerator. We model the fourth feature as a vector $\bx$ of size $11,000,000$. The vector is non-negative.

\item {\tt Meme} \cite{meme}. This dataset includes memes from the {\tt memetracker.org}. We model the vector $\bx$ as the lengths of memes. Each coordinate of $\bx$ can be thought as the number of words of a specific meme.  The dimension of $\bx$ is 210,999,824. 

\item {\tt Hudong}  \cite{konect:zhishi}. There are ``related to'' links between articles of the Chinese online encyclopaedia Hudong.\footnote{\url{http://www.hudong.com/}} This dataset contains about 2,452,715 articles, and 18,854,882 edges.  Each edge $(a, b)$ indicates that in article $a$, there is a ``related to'' link pointing to article $b$. Such links can be added or removed by users. We consider edges as a data stream, arriving in the order of editing time. Let $\bx$ be the out-degree of those articles, and $x_i$ is the number of ``related to'' links in article $i$.
This dataset will be used to test our algorithms in the streaming model where we dynamically maintain a sketch for $\bx$.

\end{itemize}



\paragraph{Measurements}
We measure the effectiveness of the tested algorithms by the tradeoffs between sketch size and the recovery quality. We also compare the  running time of these algorithms in the streaming setting.  

For \oneSR\ and \twoSR, we use $d = 9$ copies of CS/CM-matrices of dimensions $s \times n$ (see  Algorithm~ \ref{alg:l1-sketch} and Algorithm~\ref{alg:l2-sketch}). Theoretically we only need $O(\log n)$ extra words for \oneSR\ to estimate the bias, but in our implementation we use $s$ (typically much larger than $\log n$) extra words for both \oneSR\ and \twoSR, which makes it easier to compare the accuracies of \oneSR\ and \twoSR. Moreover, it also allows us to get more accurate and stable bias estimation for \oneSR.  
For \CM, \CS, \CMCU\ and \CMLCU, we set $d=10$ so that all algorithms use $10 s$ words.  We will then vary $s$ to get multiple sketch-size versus accuracy tradeoffs.


For point query we use the following two measurements: (1) average error $\frac{1}{n}\norm{\bx - \hat{\bx}}_1$, and (2) maximum error $\norm{\bx - \hat{\bx}}_\infty$.  Recall that $\hat{\bx}$ is the approximation of $\bx$ given by the recovery scheme.




\paragraph{Computation Environments}
All algorithms were implemented in C++.
All experiments were run in a server with 32GB RAM and an Intel Xeon E5-2650 v2 8-core processor; the operating system is Red Hat Enterprise Linux 6.7. 

\subsection{Accuracy for Point Query}

\begin{figure*}[!ht]
     \centering
     \subfloat[][$b=100$, {ave.~error}]{\includegraphics[width=0.22\textwidth,height=0.20\textwidth]{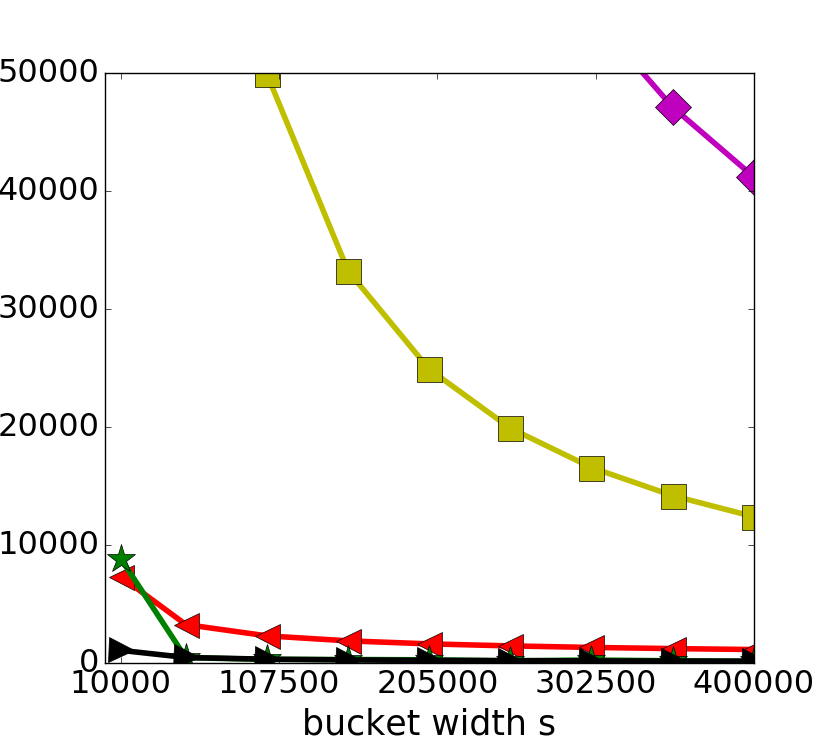}\label{fig:gauss-100-ave}}
     \subfloat[][$b=100$, {max.~error}]{\includegraphics[width=0.22\textwidth,height=0.20\textwidth]{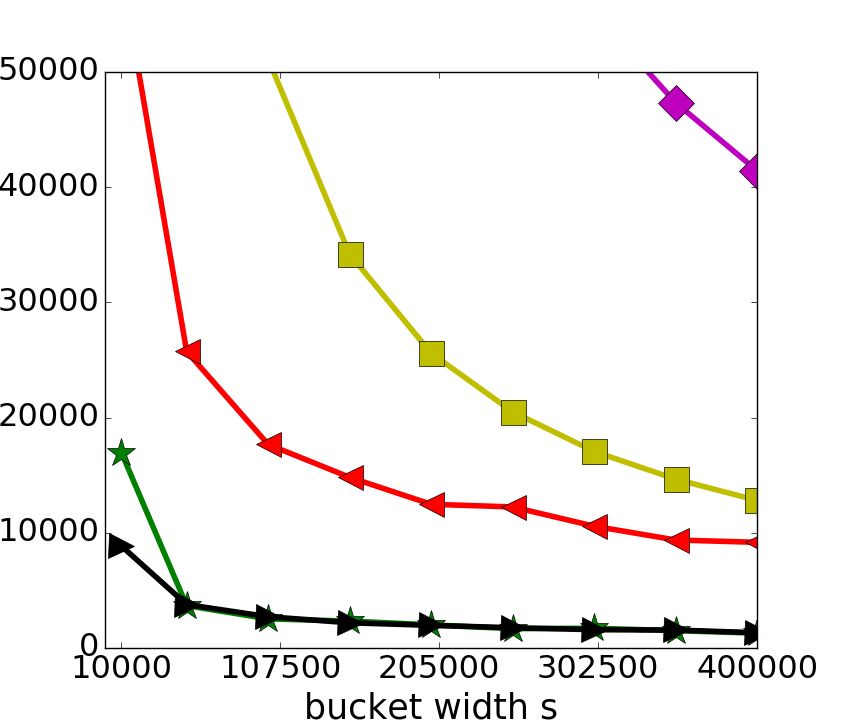}\label{fig:gauss-100-max}}
     \subfloat[][$b=500$, {ave.~error}]{\includegraphics[width=0.22\textwidth,height=0.20\textwidth]{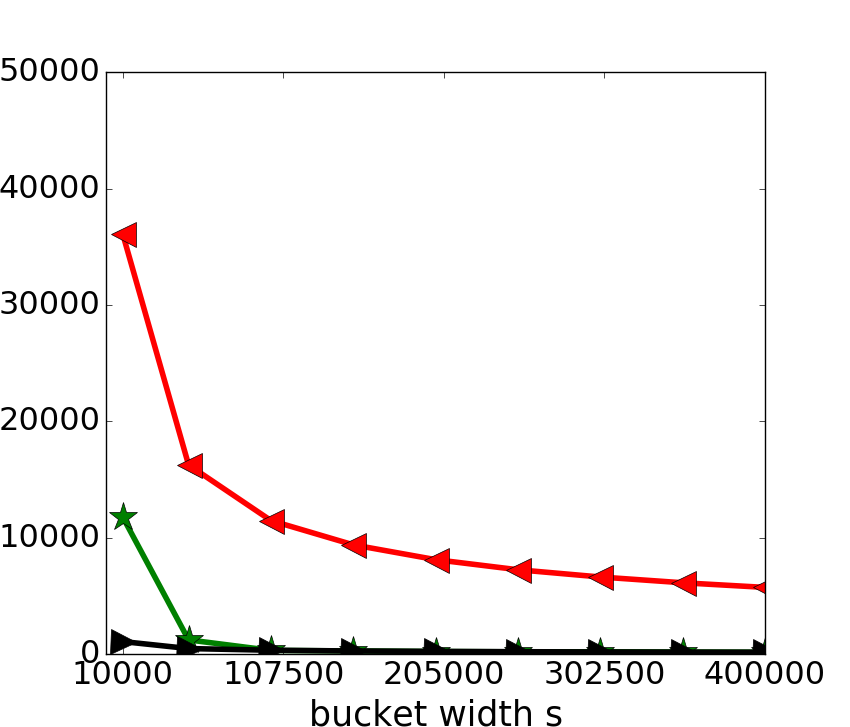}\label{fig:gauss-500-ave}}
     \subfloat[][$b=500$, {max.~error}]{\includegraphics[width=0.22\textwidth,height=0.20\textwidth]{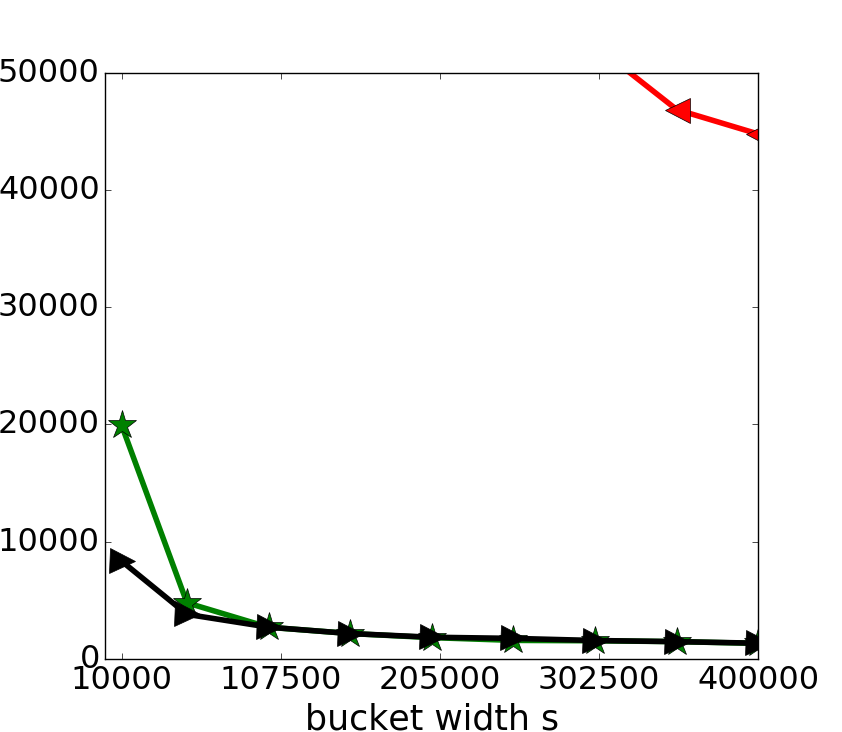}\label{fig:gauss-500-max}}
     \subfloat[][Legend]{\includegraphics[width=0.1\textwidth]{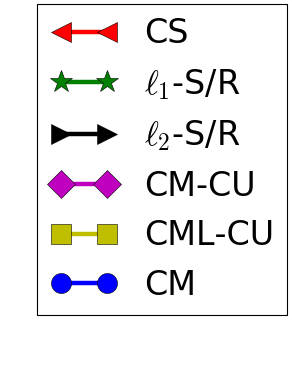}\label{fig:legend}}
     \caption{{\tt Gaussian} dataset; $n = 500,000,000$ and $\sigma=15$. Some curves for \CM, \CMCU, \CMLCU\ cannot be presented since the errors are too large}
     \label{fig:gauss}
\end{figure*}

\begin{figure}[!ht]
     \centering
     \subfloat[][{Average error}]{\includegraphics[width=0.25\textwidth,height=0.20\textwidth]{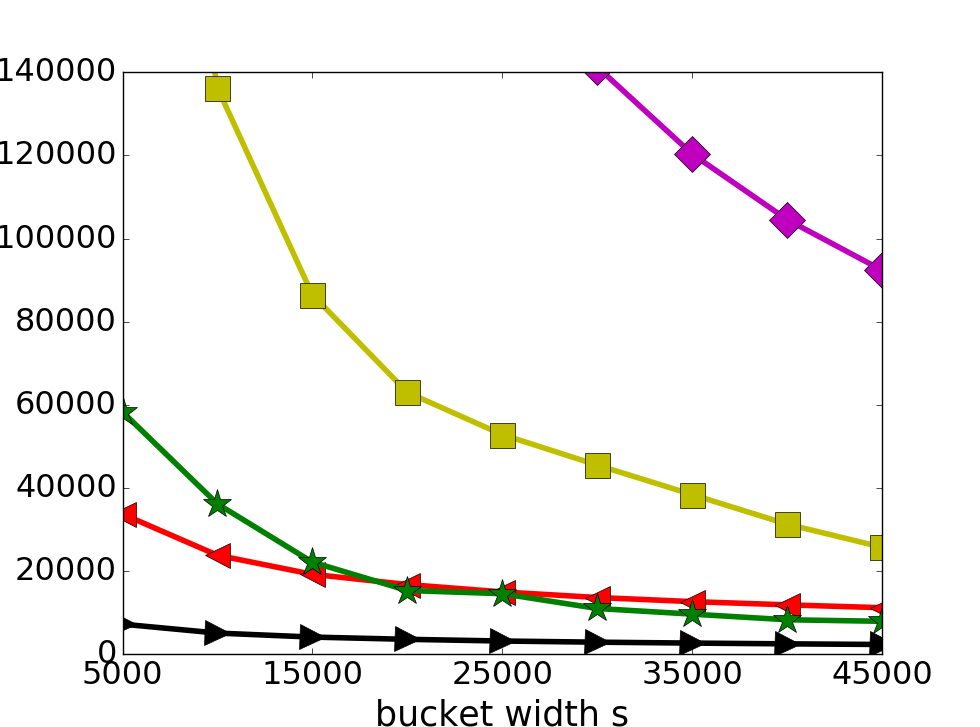}\label{fig:wiki-ave}}
     \subfloat[][{Maximum error}]{\includegraphics[width=0.25\textwidth,height=0.20\textwidth]{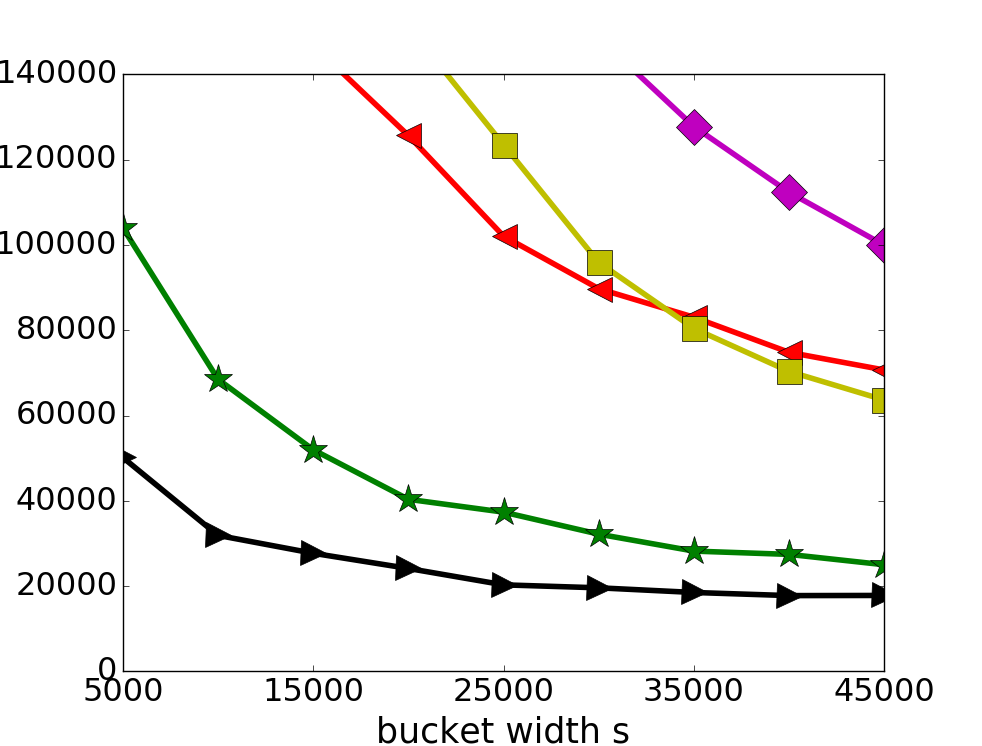}\label{fig:wiki-inf}}
     \caption{{\tt Wiki} dataset; $n = 3,513,600$. The curve for \CM\ cannot be presented since the errors are too large}
     \label{fig:wiki}
\end{figure}

\ifdefined\WORLDCUP
\begin{figure}[!ht]
     \centering
     \subfloat[][{Average error}]{\includegraphics[width=0.25\textwidth,height=0.20\textwidth]{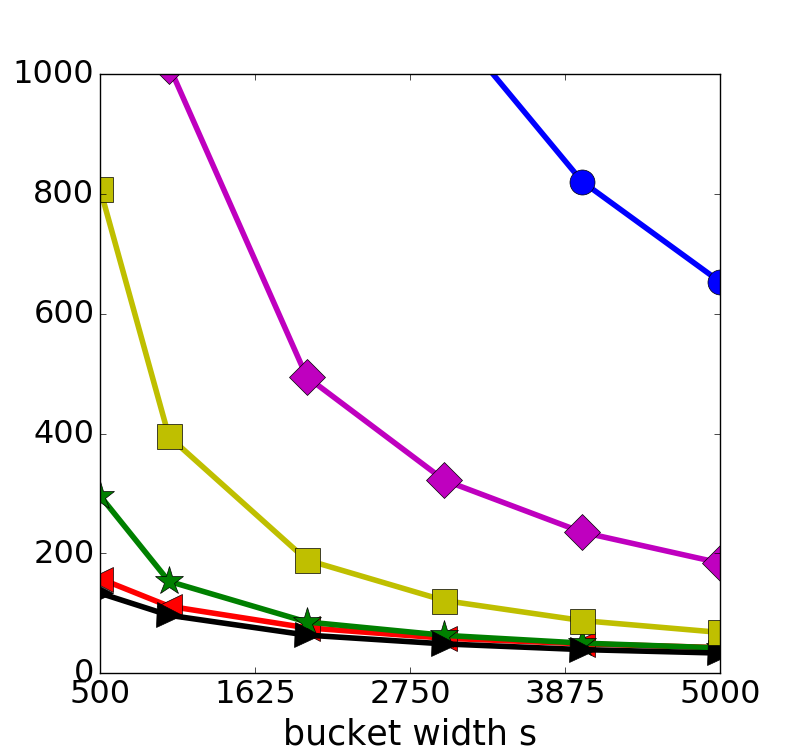}\label{fig:wc-ave}}
     \subfloat[][{Maximum error}]{\includegraphics[width=0.25\textwidth,height=0.20\textwidth]{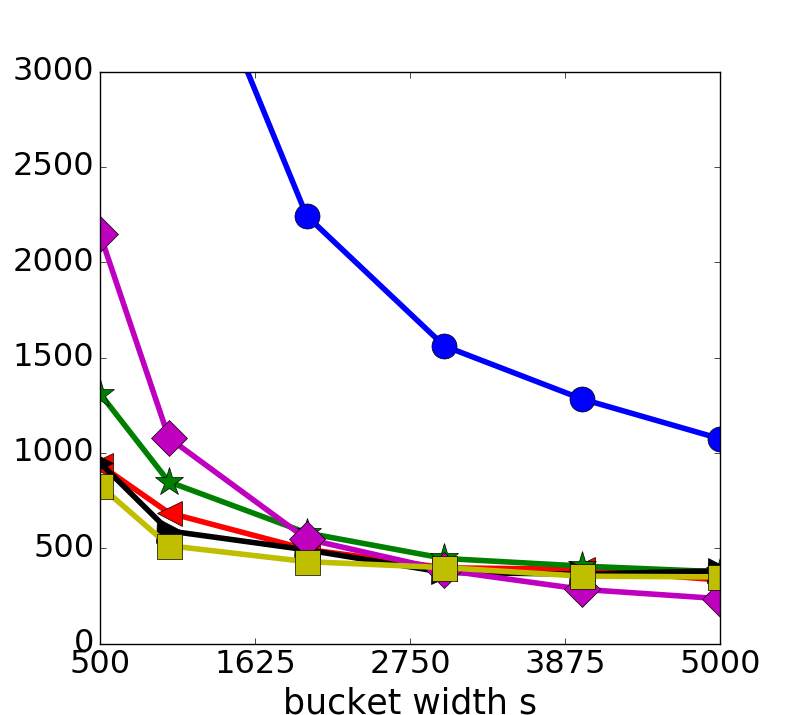}\label{fig:wc-inf}}
     \caption{{\tt WorldCup} dataset; $n = 86,400$}
     \label{fig:wc}
\end{figure}
\fi

\begin{figure}[!ht]
     \centering
     \subfloat[][{Average error}]{\includegraphics[width=0.25\textwidth,height=0.20\textwidth]{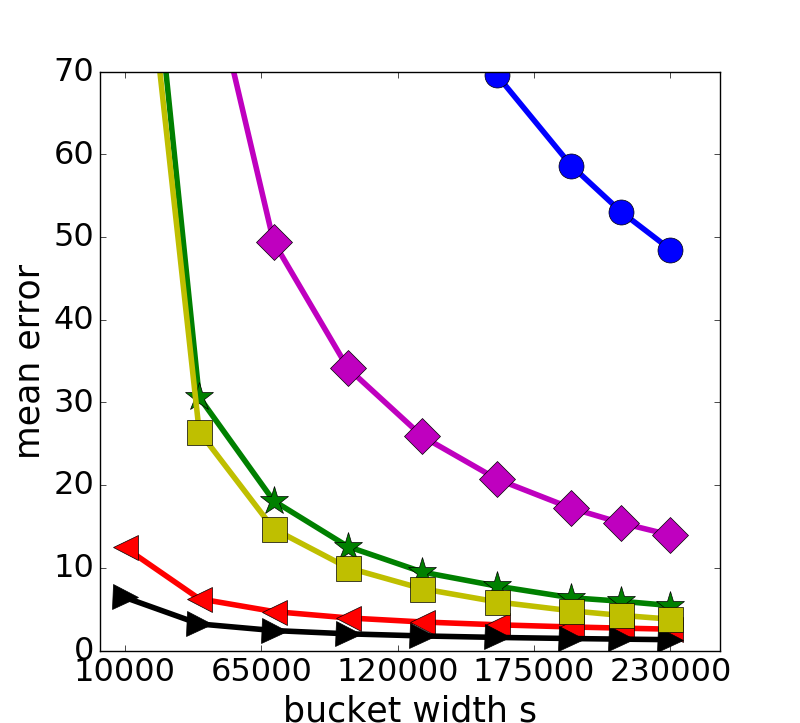}\label{fig:higgs-4-mean}}
     \subfloat[][{Maximum error}]{\includegraphics[width=0.25\textwidth,height=0.20\textwidth]{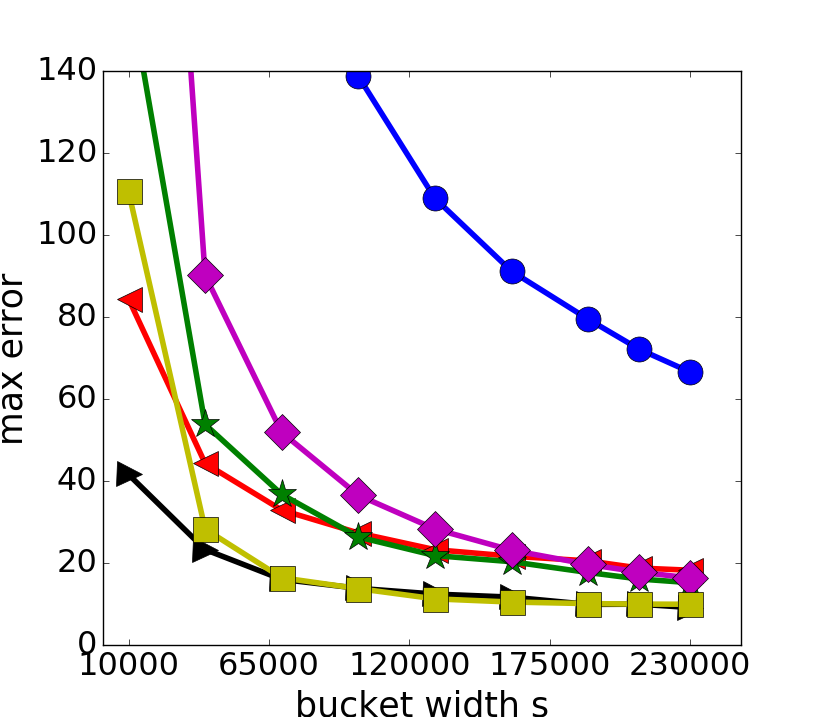}\label{fig:higgs-4-max}}
     \caption{{\tt Higgs} dataset; $n = 11,000,000$.}
     \label{fig:higgs}
\end{figure}

\begin{figure}[!ht]
     \centering
     \subfloat[][{Average error}]{\includegraphics[width=0.25\textwidth,height=0.20\textwidth]{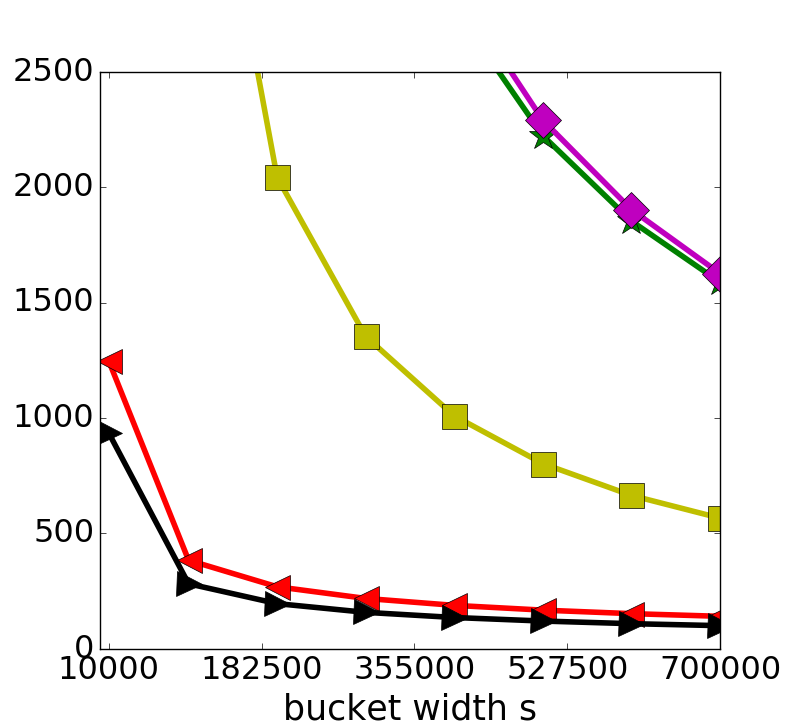}\label{fig:meme-mean}}
     \subfloat[][{Maximum error}]{\includegraphics[width=0.25\textwidth,height=0.20\textwidth]{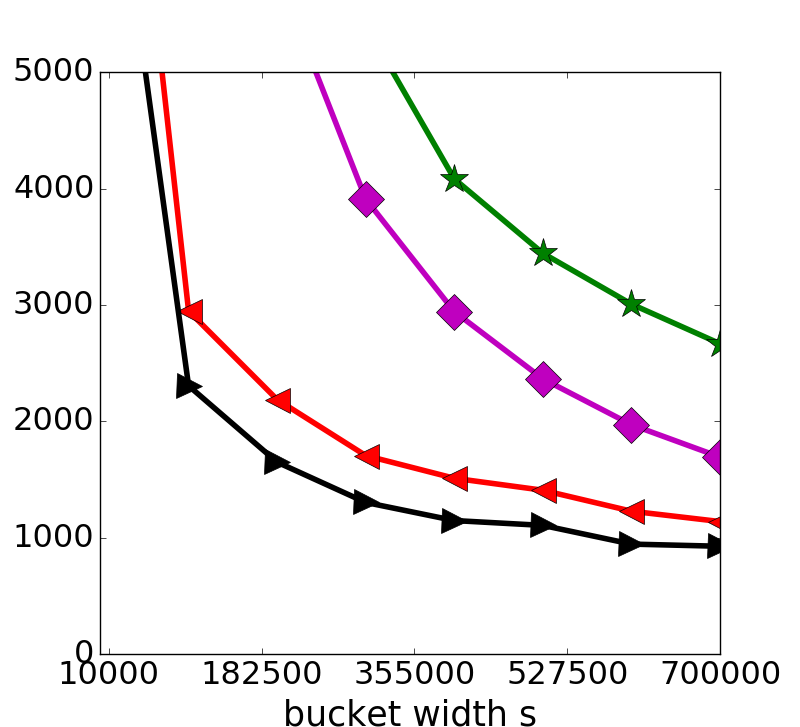}\label{fig:meme-max}}
     \caption{{\tt Meme} dataset; $n = 210,999,824$. Some curves for \CM\ and \CMLCU\ cannot be presented since the errors are too large}
     \label{fig:meme}
\end{figure}



\begin{figure*}[!ht]
     \centering
     \subfloat[][{Average error}]{\includegraphics[width=0.22\textwidth,height=0.20\textwidth]{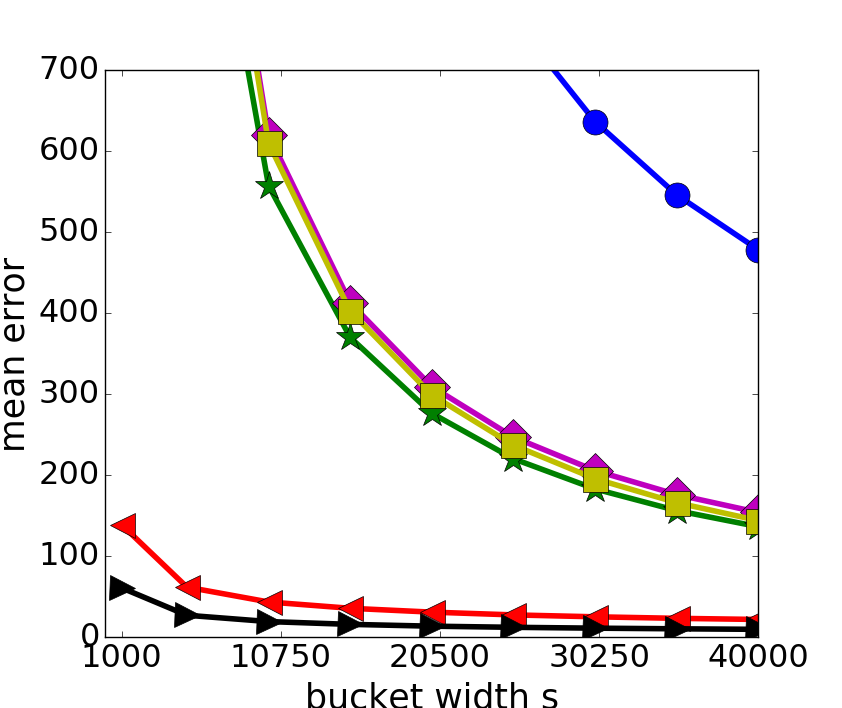}\label{fig:hudong-ave}}
     \subfloat[][{Maximum error}]{\includegraphics[width=0.22\textwidth,height=0.20\textwidth]{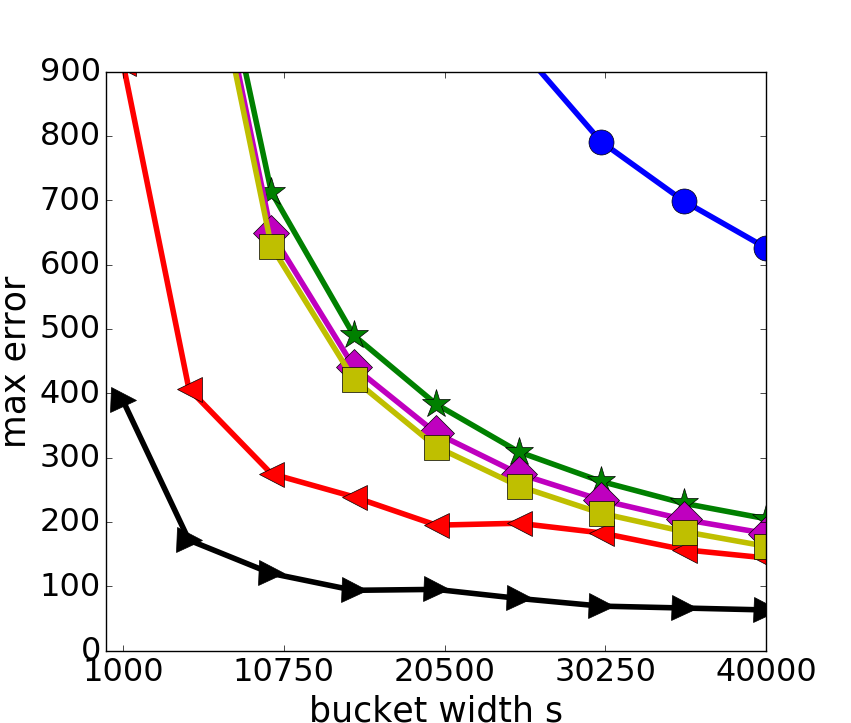}\label{fig:hudong-inf}}
     \subfloat[][Update time]{\includegraphics[width=0.22\textwidth,height=0.20\textwidth]{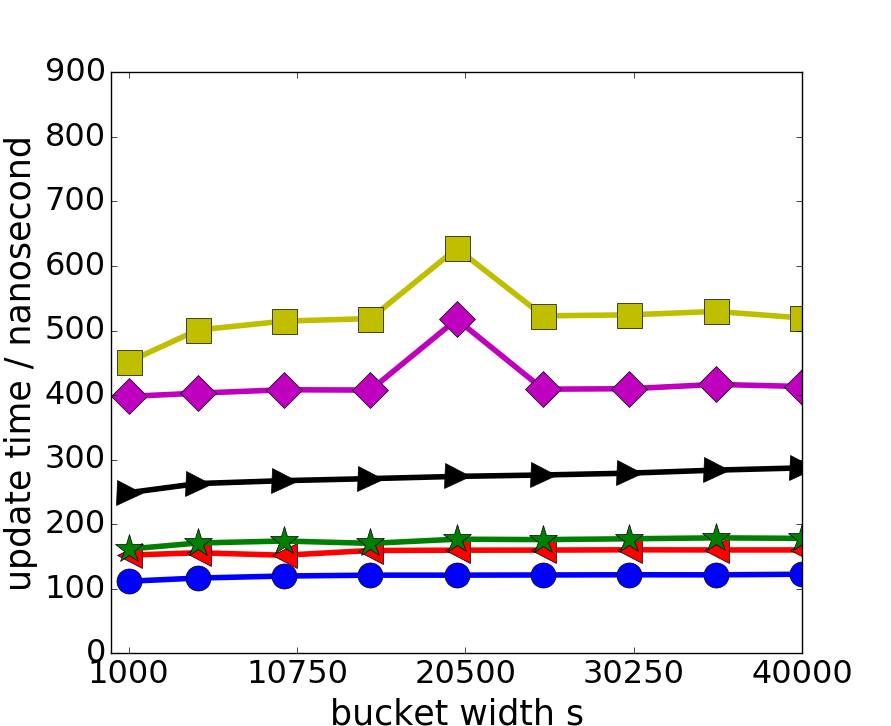}\label{fig:update-time}}
     \subfloat[][Query time]{\includegraphics[width=0.22\textwidth,height=0.20\textwidth]{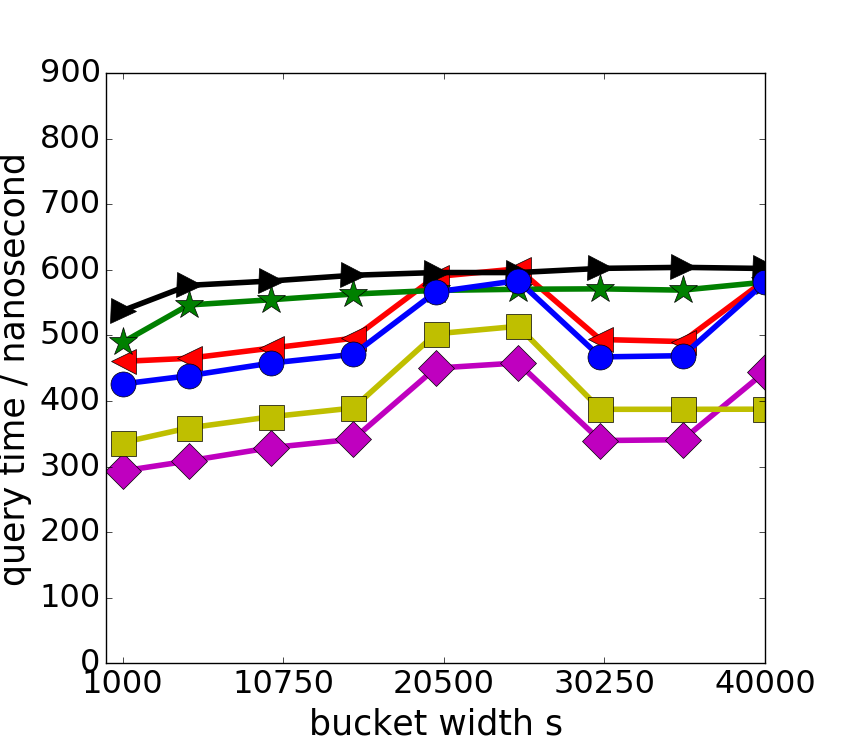}\label{fig:point-query}}
      \subfloat[][Legend]{\includegraphics[width=0.1\textwidth]{legend}\label{fig:legend}}
     \caption{{\tt Hudong} dataset; $n = 2,232,285$, there are $18,854,882$ updates in total}
     \label{fig:hudong}
\end{figure*}

\begin{figure}[!ht]
     \centering
     \subfloat[][{Average error}]{\includegraphics[width=0.25\textwidth,height=0.20\textwidth]{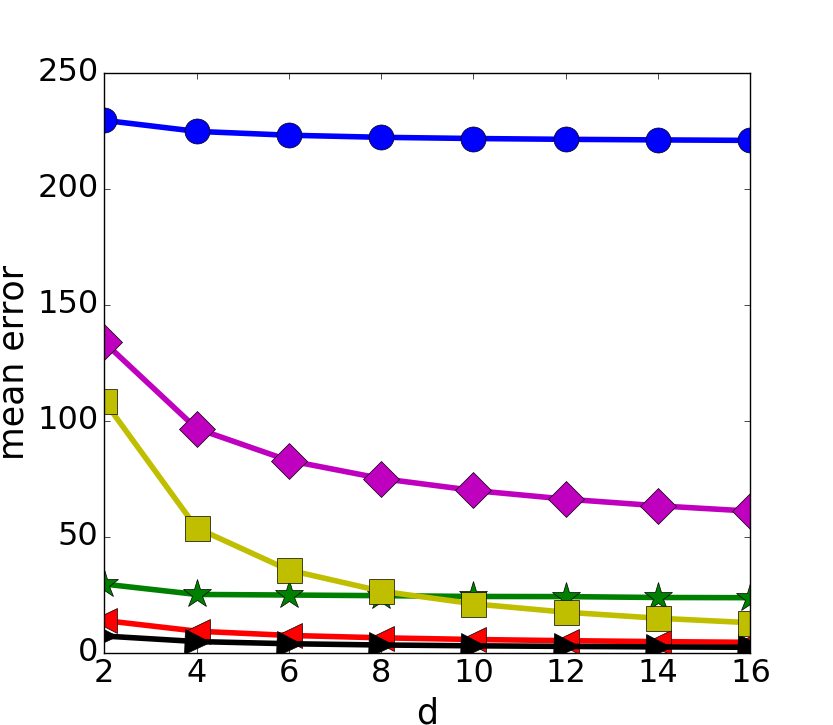}\label{fig:higgs-4-mean-d}}
     \subfloat[][{Maximum error}]{\includegraphics[width=0.25\textwidth,height=0.20\textwidth]{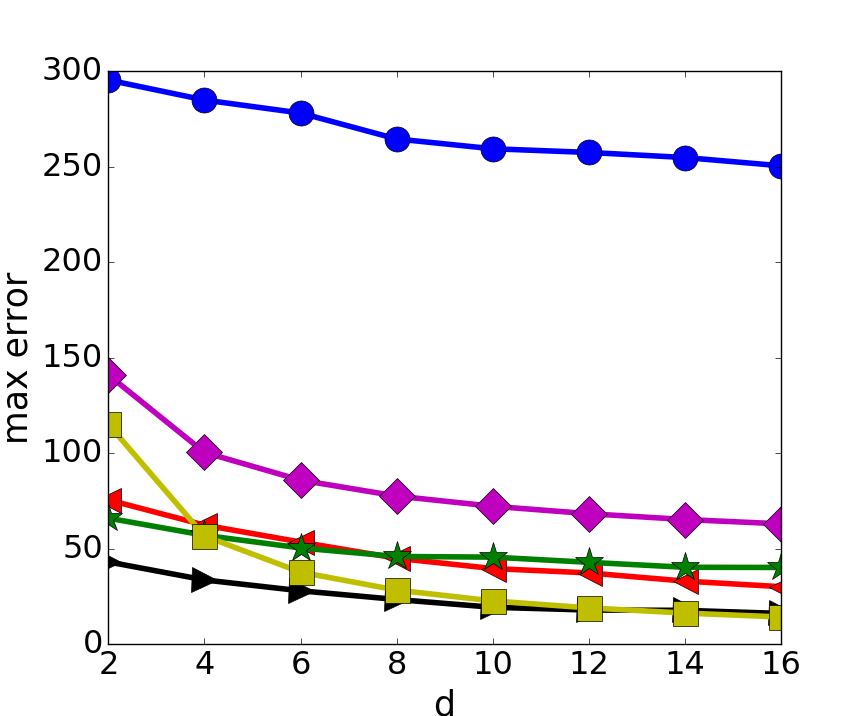}\label{fig:higgs-4-max-d}}
     \caption{{\tt Higgs} dataset for fixed $s$; $n = 11,000,000$. We fix $s = 50000$ and vary $d$. The depth $d$ here is for \oneSR\ and \twoSR; for \CS, \CM, \CMCU\ and \CMLCU, the depth is $d+1$.}
     \label{fig:higgs-fix-s}
\end{figure}

\ifdefined\FULL
\begin{figure*}[h]
     \centering
     \subfloat[][{Average error}]{\includegraphics[width=0.22\textwidth,height=0.20\textwidth]{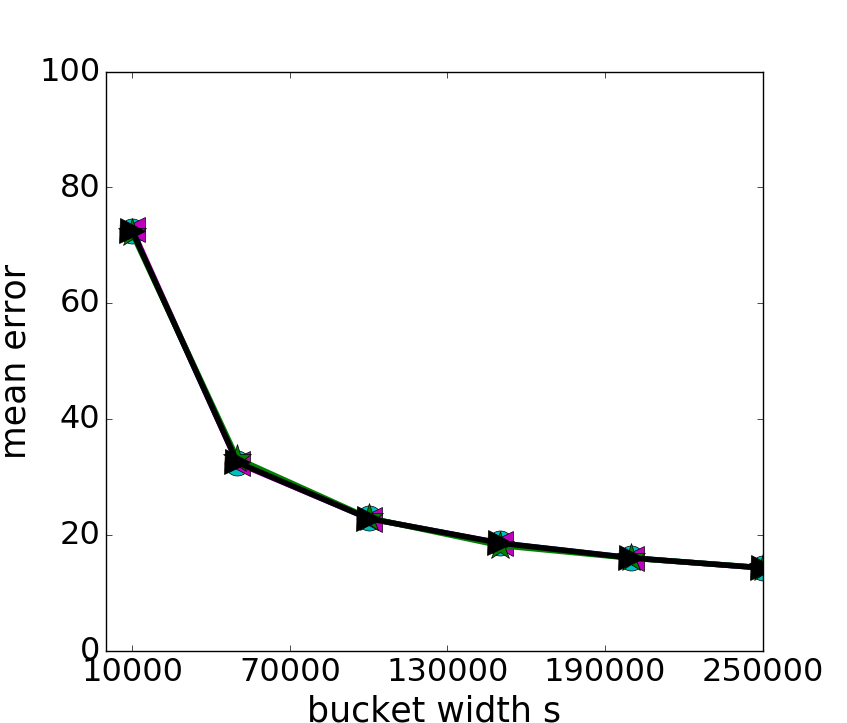}\label{fig:no-shift-mean}}
     \subfloat[][{Maximum error}]{\includegraphics[width=0.22\textwidth,height=0.20\textwidth]{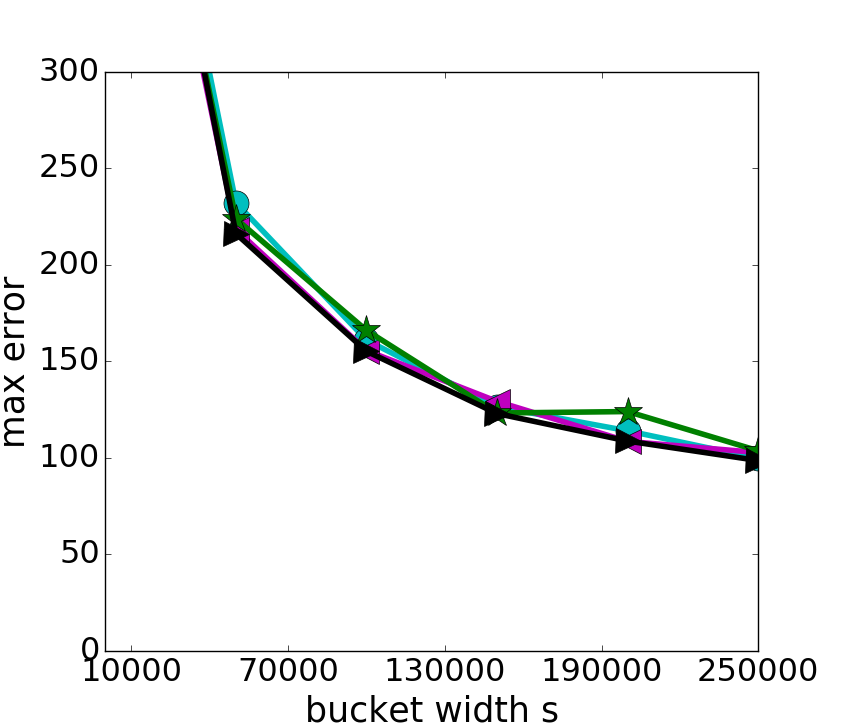}\label{fig:no-shift-max}}
     \subfloat[][{Average error}]{\includegraphics[width=0.22\textwidth,height=0.20\textwidth]{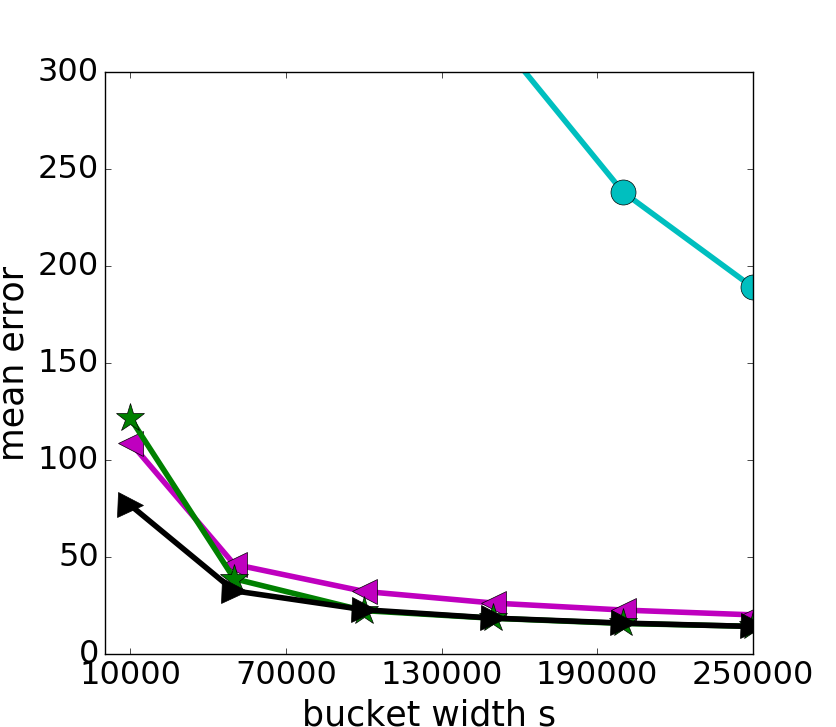}\label{fig:shift-mean}}
     \subfloat[][{Maximum error}]{\includegraphics[width=0.22\textwidth,height=0.20\textwidth]{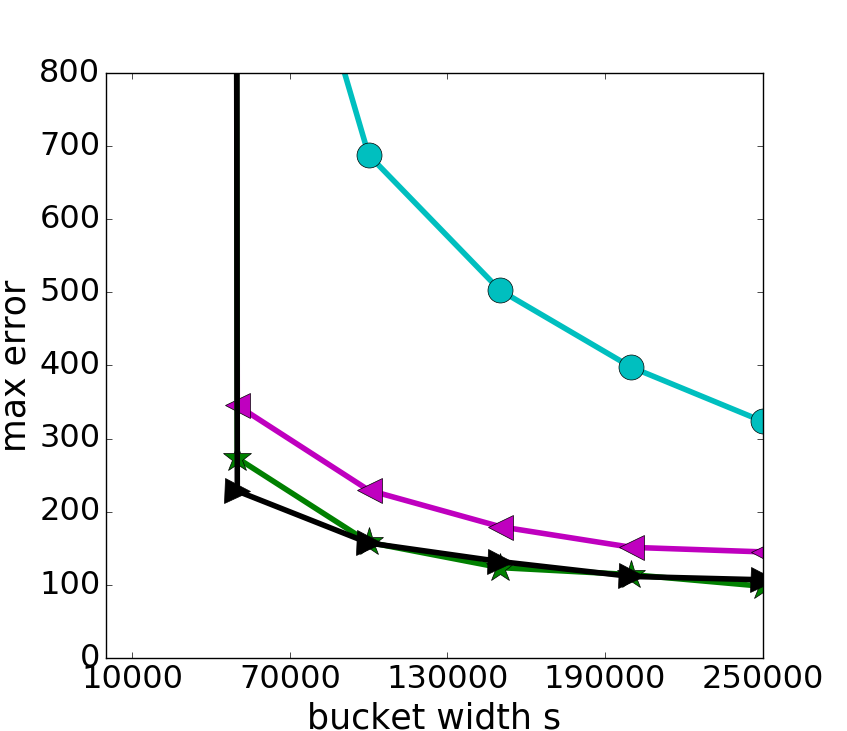}\label{fig:shift-max}}
     \subfloat[][Legend]{\includegraphics[width=0.1\textwidth]{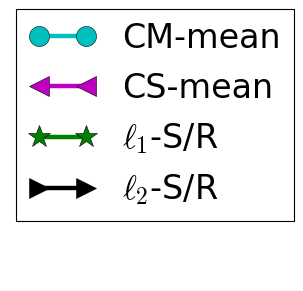}}
     \caption{{\tt Gaussian-2} dataset; Fig. \ref{fig:no-shift-mean}-\ref{fig:no-shift-max}, the dataset is not shifted.  Fig. \ref{fig:shift-mean}-\ref{fig:shift-max}, $500$ entries are shifted by $100,000$}
     \label{fig:gauss-2}
\end{figure*}
\begin{figure*}[h]
     \centering
     \subfloat[][{Average error}]{\includegraphics[width=0.22\textwidth,height=0.20\textwidth]{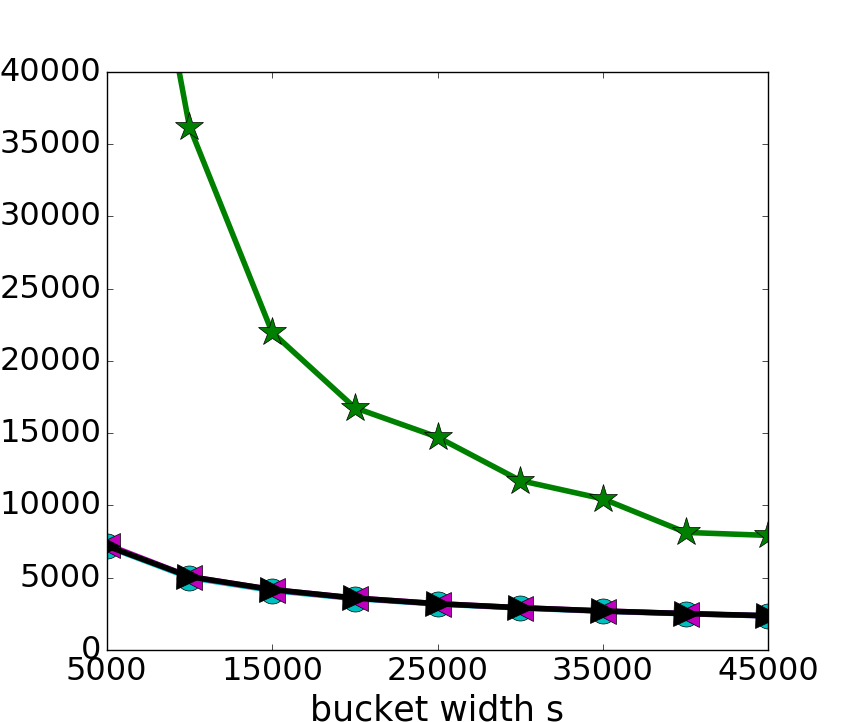}\label{fig:wiki-simple-mean}}
     \subfloat[][{Maximum error}]{\includegraphics[width=0.22\textwidth,height=0.20\textwidth]{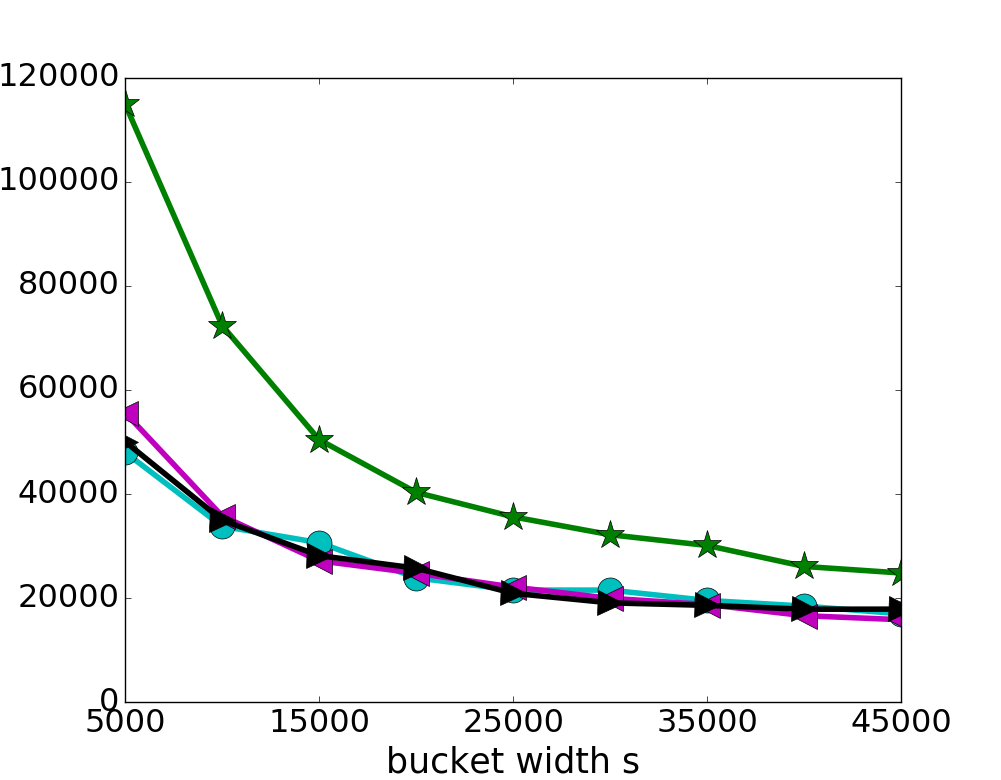}\label{fig:wiki-simple-max}}
     \caption{{\tt Wiki} dataset}
     \label{fig:wiki-simple}
\end{figure*}

\fi

{\tt Gaussian} dataset with $n = 500,000,000$.
Figure~\ref{fig:gauss-100-ave} and Figure~\ref{fig:gauss-100-max}  show the average and maximum errors of \oneSR, \twoSR, \CM, \CS, \CMCU\ and \CMLCU\ respectively on {\tt Gaussian} dataset with $n=500$ million, $\sigma=15$ and $b=100$.

First note that \oneSR\ and \twoSR\ have similar average/maximum errors when we increase $s$. An explanation for this phenomenon is that in \twoSR\ we use random signs $+1, -1$ to reduce/cancel the noise (contributed by colliding coordinates) in each hashing bucket, while in \oneSR\ we do not. But in {\tt Gaussian} the ``perturbation'' of each $x_i$ around the bias is symmetric, and thus both algorithms achieve good cancellations.  When $s$ is small,  the error of \twoSR\ is slightly smaller than that of \oneSR, this might because \oneSR\ can not estimate the bias accurately. 
On the other hand, both \oneSR\ and \twoSR\ outperform other algorithms significantly. As a comparison, the errors of \oneSR\ and \twoSR\ are less than $1/5$ of \CS, $1/20$ of \CMLCU, $1/50$ of \CMCU\ and $1/200$ of \CM.


In Figure~\ref{fig:gauss-500-ave} and Figure~\ref{fig:gauss-500-max}, we increase the value of $b$ to $500$. As we can observe from those figures that the average and maximum errors of \oneSR\ and \twoSR\ are not affected by the value of $b$, which can be fully predicted from our theoretical results. On the contrary, the errors of \CM, \CS, \CMCU\ and \CMLCU\ increase significantly when we increase $b$.



 \smallskip

{\tt Wiki} dataset.
Figure~\ref{fig:wiki} shows the accuracies of different algorithms on {\tt Wiki}. We have observed that when varying the sketch size, \twoSR\ always achieves the best recovery quality. For example, when sketch size is $s = 20,000$,  the average error of \twoSR\ is smaller than $1/10$ of the average errors of other algorithms. For average error, \oneSR\ and \CS\ perform similarly but the maximum error of \CS\ is typically $2+$ times larger than that of \oneSR. The performance of \CM, \CMCU\ and \CMLCU\ are much worse than \oneSR\ and \twoSR.
\smallskip

\ifdefined\WORLDCUP
{\tt WorldCup} dataset.
Figure~\ref{fig:wc} shows the accuracies of different algorithms on {\tt WorldCup}. While \twoSR\ still achieves the smallest  average error, \CS\ and \oneSR\ follow closely.  Again \CM, \CMCU\ and \CMLCU\ perform significantly worse than others. For maximum error, \CS, \CMCU, \CMLCU\, \oneSR\ and \twoSR\ have similar errors; \CM\ gives significantly (typically $4+$ times) larger errors than other algorithms.
\smallskip
\fi

{\tt Higgs} dataset.
Figure~\ref{fig:higgs} shows the accuracies of different algorithms on {\tt Higgs}. It can be observed that for average error, \twoSR\ again achieves the smallest error. The average error of \CS\ is typically larger than that of \twoSR\ and is much smaller than that of other algorithms.
For maximum error, \CMLCU\ has similar accuracy as \twoSR\ for large $s$. The maximum errors of all other algorithms are larger than that of \twoSR. \CM\ again has the worst performance.
\smallskip

{\tt Meme} dataset.
Figure~\ref{fig:meme} shows the accuracies of different algorithms on {\tt Meme}. We can again observe that \twoSR\ achieves the best recovery quality. The errors of \CS\ are about $30\%$ larger than that of \twoSR. Both \twoSR\ and \CS\ outperform other algorithms significantly.

\subsection{Effects of Sketch Depth}
To see how the sketch depth $d$ affects the accuracy of the sketch, we conduct experiments as follows: we fix the sketch size $s$ and vary the sketch depth $d$. We only present the results for {\tt Higgs} and similar results can be observed in other datasets.

It can be observed from Figure \ref{fig:higgs-fix-s} that for all algorithms we tested, increasing $d$ will improve the accuracy. One can also observe that \CMLCU\ is more sensitive to the value of $d$ than other algorithms. In terms of accuracy, \twoSR\ still outperforms other algorithms and for the maximum error, \CMLCU\ follows closely when $d$ is large.

\ifdefined\FULL
\subsection{Comparisons with Mean Heuristics}
\label{sec:ell-mean}
We also compare our algorithms with \ellmeanOne\ and \ellmeanTwo. In Figure \ref{fig:no-shift-mean}-\ref{fig:no-shift-max}, we use the dataset whose entries are sampled from $\mathcal{N}(100, 15^2)$. It can be observed that all algorithms have similar performance and this is because all of \oneSR, \twoSR, \ellmeanOne\ and \ellmeanTwo\ can estimate the bias ($b=100$) well.  In the dataset used in Figure \ref{fig:shift-mean}-\ref{fig:shift-max}, we shift $500$ entries by $100,000$. A direct consequence is that the mean of the vector is no longer an accurate estimation of the bias. It can be observed that errors of both \ellmeanOne\ and \ellmeanTwo\ increase significantly.

We also conduct experiments on {\tt Wiki} dataset, one can observe that \twoSR, \ellmeanOne\ and \ellmeanTwo\ have similar performance and all of them outperform \oneSR. 

\fi

\subsection{Distributed and Streaming Implementations}

As mentioned in the introduction, it is straightforward to implement our bias-aware sketches in the distributed model by making use of the linearity. Moreover, their performance in the distributed model can be {\em fully} predicted by the centralized counterparts -- the total communication will just be the number of sites times the size of the sketch, and the time costs at the sites and the coordinator will be equal to the sketching time and recovery time respectively.\footnote{Regarding the random hash functions, the coordinator can simply generate $g, h^1, \ldots, h^d : [n] \to [s]$; $r^1, \ldots, r^d : [n] \to \{-1, 1\}$ at the beginning and send to each site, which only incurs an extra of $O(\log n)$ communication on each channel and is thus negligible.} Therefore, our experiments in the centralized model can also speak for that in the distributed model.

We implemented our bias-aware sketches in the streaming model.
\ifdefined\FULL
\else
We refer readers to Algorithm 5 and 6 in the full version~\cite{CZ16e} for the details.
\fi
We have run our algorithms on the streaming dataset {\tt Hudong} where edges are added in the streaming fashion. We update the sketch at each step, and recover the entire $\hat{\bx}$ after feeding in the whole dataset.  To measure the running time, we first process the whole data stream and calculate the average update time. We then recover the whole vector and calculate the average query time.


\paragraph{Accuracy for Point Query}
Figure \ref{fig:hudong-ave} and Figure \ref{fig:hudong-inf} show that the recovery errors of \CS\ are $2+$ times larger than that of \twoSR. The others algorithms have even larger errors. In both Figures the results of \CMLCU\ and \CMCU\ are very close and their curves overlap with each other. The performance of \oneSR\ is also quite similar to \CMLCU\ and \CMCU.

\paragraph{Update/Recover Running Time}
It can be seen from Figure \ref{fig:update-time} and Figure \ref{fig:point-query} that all of the six tested algorithms have similar processing time per update and per point query.  The time cost per update of $\ell_1$-S/R is about 50\% more than \CM, and that of $\ell_2$-S/R is within a factor of 2 of \CS.  We thus conclude that the additional components (such as the Bias-Heap) used in \oneSR\ and \twoSR\ only generate small overheads.  

\subsection{Summary of Experimental Results}
We now summarize our experimental results.  We have observed that in terms of recovery quality, \oneSR\ strictly outperforms \CM, and \twoSR\ strictly outperforms \CS. In general \twoSR\ is much better than \oneSR, especially when the noise around the bias is not symmetric.  Note that this is similar to the phenomenon that the error of \CS\ is almost always smaller than that of \CM\ in practise, and is consistent to the theoretical fact that if $n \gg k$, and the tail coordinates of $\by = \bx - \beta^{(n)}$ follows some long tail distribution, than the error $\frac{1}{k} \err{1}(\by)$ is much larger than $\frac{1}{\sqrt{k}} \err{2}(\by)$.

In almost all datasets we have tested, \twoSR\ outperforms \CMLCU\ and \CMCU, the latter two are considered as improved versions of the Count-Min sketch.

The sketch depth $d$ also affects the accuracy of a sketch. Larger $d$ leads to better performance. It is also observed that some algorithms (e.g. \CMLCU) are more sensitive to $d$ than others.

As for running time (update/query), the differences between \oneSR, \twoSR, \CS, \CM, \CMCU\ and \CMLCU\ are not significant. The overhead introduced by the components  used to estimate the bias is fairly low in both \oneSR\ and \twoSR.



\section{Conclusion}
\label{sec:conclude}

In this paper we formulated the bias-aware sketching and recovery problem, and proposed two algorithms that strictly generalize the widely used Count-Sketch and Count-Median algorithms. Our bias-aware sketches, due to their linearity, can be easily implemented in the streaming and distributed computation models. We have also verified their effectiveness experimentally, and showed the advantages of our bias-aware sketches over Count-Sketch, Count-Median and the improved versions of Count-Min in both synthetic and real-world datasets.

\section{ACKNOWLEDGMENT}
Jiecao Chen and Qin Zhang are supported in part by NSF CCF-1525024 and IIS-1633215.

\bibliographystyle{abbrv}
\bibliography{paper}


\end{document}